\documentclass[english]{article}
\usepackage[T1]{fontenc}
\usepackage[latin1]{inputenc}
\usepackage{geometry}
\usepackage{float}
\usepackage{mathrsfs}
\usepackage{amsmath}
\usepackage{amssymb}
\usepackage{graphicx}
\usepackage{subfigure}

\def\sp{\textsc{s\&p}500\xspace}
\def\ou{Ornstein-Uhlenbeck\xspace}

\makeatletter

\floatstyle{ruled}
\newfloat{algorithm}{tbp}{loa}
\providecommand{\algorithmname}{Algorithm}
\floatname{algorithm}{\protect\algorithmname}

\usepackage{amsthm}

\usepackage{mathrsfs}

\usepackage{amsfonts}

\usepackage{epsfig}

\usepackage{bm}

\usepackage{mathrsfs}

\usepackage{enumerate}

\@ifundefined{definecolor}{\@ifundefined{definecolor}
 {\@ifundefined{definecolor}
 {\usepackage{color}}{}
}{}
}{}

\usepackage{subfig}\usepackage[all]{xy}

\newcommand{\bbR}{\mathbb R}

\newtheorem{theorem}{Theorem}[section]
\newtheorem{lem}{Lemma}[section]
\newtheorem{rem}{Remark}[section]
\newtheorem{prop}{Proposition}[section]
\newtheorem{cor}{Corollary}[section]
\newtheorem{asn}{Assumption}[section]
\newcounter{hypA}

\newcommand{\Exp}{\mathbb{E}}

\newcommand{\bbE}{\mathbb{E}}
\newcommand{\bbP}{\mathbb{P}}

\newcommand{\cA}{\mathcal{A}}

\newcommand{\cO}{\mathcal{O}}

\newcommand{\cB}{\mathcal{B}}

\newcommand{\cU}{\mathcal{U}}
\newcommand{\cV}{\mathcal{V}}

\newcommand{\hU}{\widehat{U}}
\newcommand{\heta}{\widehat{\eta}}

\newcommand{\cC}{\mathcal{C}}

\usepackage{xspace}
\usepackage{tabu}
\usepackage{booktabs}
\def\ess{ESS\xspace}
\def\gbm{GBM\xspace}
\def\mlpf{MLPF\xspace}
\def\mse{MSE\xspace}
\def\nlm{NLM\xspace}
\def\ou{OU\xspace}
\def\pf{PF\xspace}
\def\sde{SDE\xspace}
\def\bbR{\mathbb{R}}
\def\calL{\mathcal{L}}
\def\calN{\mathcal{N}}
\def\sp{S\&P~500\xspace}
\DeclareMathOperator\var{var}

\usepackage{babel}\date{}

\usepackage{babel}

\makeatother

\usepackage{babel}
\begin{document}

\begin{center}

{\Large \textbf{Multilevel Particle Filters}}

\vspace{0.5cm}

BY AJAY JASRA$^{1}$, KENGO KAMATANI$^{2}$, KODY J. H. LAW$^{3}$ \& YAN ZHOU$^{1}$

{\footnotesize $^{1}$Department of Statistics \& Applied Probability,
National University of Singapore, Singapore, 117546, SG.}
{\footnotesize E-Mail:\,}\texttt{\emph{\footnotesize staja@nus.edu.sg; stazhou@nus.edu.sg}}\\
{\footnotesize $^{2}$Graduate School of Engineering Science, Osaka University, Osaka, 565-0871, JP.}
{\footnotesize E-Mail:\,}\texttt{\emph{\footnotesize kamatani@sigmath.es.osaka-u.ac.jp}}\\
{\footnotesize $^{3}$Computer Science and Mathematics Division,
Oak Ridge National Laboratory Oak Ridge, TN, 37831, USA.}\\
{\footnotesize E-Mail:\,}\texttt{\emph{\footnotesize kodylaw@gmail.com}}
\end{center}

\begin{abstract}
In this paper 
the filtering of partially observed diffusions, with discrete-time observations, is considered. 
It is assumed that only biased approximations of the diffusion can be obtained, for choice of
an accuracy parameter indexed by $l$.
A multilevel estimator 
is proposed, consisting of a telescopic sum of increment estimators
associated to the successive levels.
The work associated to $\cO(\varepsilon^2)$ mean-square error between the multilevel estimator
and average with respect to the filtering distribution is shown to scale optimally, for example 
as $\cO(\varepsilon^{-2})$ for optimal rates of convergence of the underlying diffusion approximation.
The method is illustrated on some toy examples as well as estimation of interest rate based on real 
S\&P 500 stock price data. 
\end{abstract}

\section{Introduction}

Problems which involve 
continuum fields are typically discretized before they
are solved numerically.  Finer resolution solutions are more expensive to compute than coarse resolution ones.
Often such discretizations naturally give rise to resolution hierarchies, for example nested meshes.
Successive solution on refined meshes can be utilized to mitigate the number of necessary solves at the 
finest resolution.  For solution of linear systems, the coarsened systems are solved as pre-conditioners 
within the framework of iterative linear solvers in order to reduce the condition number, and hence the 
number of necessary iterations, at the fine resolution.  This is the principle of multi-grid methods \cite{briggs2000multigrid}.

In the context of Monte Carlo methods, 
a telescoping sum of correlated differences at successive refinement levels 
can be utilized so that the bias of the resulting multilevel estimator is determined by the finest 
level but the variance is given by the sum of the variances of the increments.    
The decay in the variance of the increments of finer levels means that the number of samples 
required to reach a given error tolerance is also reduced for finer levels.  This can then be optimized to balance 
the extra per-sample cost at the finer levels \cite{heinrich2001multilevel, Giles08, giles_acta}.

Inference tends to be more complicated, especially in a Bayesian context,
as the posterior measure often concentrates strongly with respect to the prior.  
Therefore, simple Monte Carlo strategies involving ratios of likelihood-weighted integrals 
tend to converge slowly and be inefficient.  Indeed, in extreme cases all the weight may 
concentrate on a single sample:
this is referred to as weight degeneracy.  In the case in 
which data arrives sequentially online, as considered here, this phenomenon compounds,
and degeneracy is unavoidable without a resampling mechanism (see e.g.~\cite{delm:04,doucet_johan}).  
If resampling is performed from time to time, and if the data and underlying diffusion are sufficiently regular, then 
degeneracy can be avoided and even time-uniform convergence is possible \cite{delm:04, delm_guionnet2001}.

The natural and yet challenging extension of the multilevel Monte Carlo (MLMC) 
framework to inference problems has recently been pioneered by the works 
\cite{hoangmlmcmc2013, scheichlmlmcmc2013, oursmlsmc2015, hoel2015multilevel},
but, to the best knowledge of the authors, rigorous results for consistent filtering, 
via the particle filter, have yet to be obtained.
In this article, 
the context of a partially observed diffusion is considered, with observations in discrete time; 
this will be detailed 
explicitly in the next section.

In the context of filtering, one difficulty is the nonlinearity of the update,
which precludes the construction of unbiased estimators.  However, this problem was already 
addressed in \cite{oursmlsmc2015}. 
Indeed some ingenuity is required to successfully actualize the necessary resampling step
while retaining adequate correlations.  
In this paper a novel coupled resampling procedure 
is introduced, which enables
this extension of the MLMC framework to the multilevel particle filter (MLPF).
The work associated to $\cO(\varepsilon^2)$ mean-square error between the multilevel estimator
and average with respect to the filtering distribution is shown to scale optimally, for example 
as $\cO(\varepsilon^{-2})$ for optimal rate of convergence of the underlying diffusion approximation.

This new MLPF algorithm is illustrated on some toy diffusion examples, as well as a stochastic 
volatility model with real S\&P 500 stock price data.  The performance of the new algorithm easily reaches 
an order of magnitude or greater improvement in cost, and the theoretical rate is verified so that 
improvement will continue to amplify as more accurate estimates are obtained.
Furthermore, the method is very amenable to parallelization strategies, 
leaving open great potential for its use on next generation super-computers.

\section{Set Up}
\label{sec:setup}

Consider the following diffusion process:
\begin{eqnarray}
dX_t & = & a(X_t)dt + b(X_t)dW_t
\label{eq:sde}
\end{eqnarray}
with $X_t\in\mathbb{R}^d$, $t\geq 0$ and $\{W_t\}_{t\in[0,T]}$ a Brownian motion of appropriate dimension. 
The following assumptions will be made on the diffusion process.
\begin{asn}[SDE properties] The coefficients $a\in C^2(\mathbb{R}^d;\mathbb{R}^d), b\in C^2(\mathbb{R}^d;\mathbb{R}^d\otimes\mathbb{R}^d)$. 
Also, $a$ and $b$ satisfy 
\begin{itemize}
\item[{\rm (i)}] {\bf uniform ellipticity}: $b(x)b(x)^T$ is uniformly positive definite;
\item[{\rm (ii)}] {\bf globally Lipschitz}:
there is a $C>0$ such that 
$|a(x)-a(y)|+|b(x)-b(y)| \leq C |x-y|$ 
for all $x,y \in \bbR^d$; 
\item[{\rm (iii)}] {\bf boundedness}: $\bbE |X_0|^p < \infty$ for all $p \geq 1.$
\end{itemize}
\label{asn:diff}
\end{asn}
Notice that (ii) and (iii) together imply that $\bbE |X_n|^p < \infty$ for all $n$.

It will be assumed that the data are 
regularly spaced (i.e.~in discrete time) observations 
$y_1,\dots,y_{n}$, 
where $y_k \in \bbR^m$ is a realization of $Y_k$ and  
$Y_k| X_{k\delta}$ has density given by $G(y_k, x_{k\delta})$. 
For simplicity of notation let $\delta=1$ (which can always be done by rescaling time), so $X_k = X_{k\delta}$.
The joint probability density of the observations and the unobserved diffusion at the observation times 
is then
$$
\prod_{i=1}^n G(y_{i},x_{i})Q^\infty(x_{(i-1)},x_{i}), 
$$
where 
$Q^\infty(x_{(i-1)},x)$ 
is the transition density of the diffusion process as a function of $x$, 
i.e. the density of the solution
$X_1$ of Eq. \eqref{eq:sde} at time $1$
given initial condition 
$X_0=x_{(i-1)}$.

The following assumptions will be made on the observations.

\begin{asn}[Observation properties] 
There are some $c>1$ and $C>0$, such that $G$ satisfies
\begin{itemize}
\item[{\rm (i)}] {\bf boundedness}: $c^{-1} < G(y,x) < c$ for all $x \in \bbR^d$ and $y\in \bbR^m$; 
\item[{\rm (ii)}] {\bf globally Lipschitz}: for all $y \in \bbR^m$, $|G(y,x) - G(y,x')| \leq C |x-x'|$.
\end{itemize}
\label{asn:g}
\end{asn}

For $k\in\{1,\dots,n\}$, the objective is to approximate 
the target distribution
$
\pi^\infty(x_{k}|y_{1:k}),
$
which will be denoted $\heta^\infty_k$.
With a particle filter one obtains a collection of samples $\{{u}^{\infty,i}_k\}_{i=1}^N$ 
with associated weights $\{\omega^{\infty,i}_k\}_{i=1}^N$, giving rise to an empirical measure
$$
\heta_k^{\infty,N} = \sum_{i=1}^N \omega^{\infty,i}_k \delta_{{u}^{\infty,i}_k}
$$
which approximates $\heta_k^\infty$.
The particle filter works by interlacing importance sampling for the Bayesian updates incorporating 
observations, with a resampling selection step to rejuvenate the ensemble, 
and a mutation move which propagates the ensemble forward through the diffusion (e.g.~\cite{doucet_johan} and the references therein).  
It is a well-known fact that if $Q^\infty(x,\cdot)$ can be sampled from exactly, then the
particle filter achieves standard convergence rates for Monte Carlo approximation
of expectations of quantities of interest $\varphi:\cB_b(\bbR^d)$, 
the set of bounded measurable functions over $\bbR^d$
\cite{Cappe_2005} :
\begin{equation}
\bbE |{\heta}_k^{\infty,N}(\varphi) - \heta^\infty_k(\varphi)|^2 \leq C/N,
\label{eq:mc}
\end{equation}
although $C$ may behave poorly with respect to $k$ and/or $d$ \cite{delm:04, bengtsson2008curse, bickel2008sharp}. 
  In the setting considered in this paper,
 it is not possible to sample exactly from $Q^\infty(x,\cdot)$, with the exception of very simple
SDE \eqref{eq:sde}, but rather it must be approximated by some discrete time-stepping method \cite{KlPl92}.

It will be assumed that the diffusion process is 
approximated by a time-stepping method for time-step $h_l=2^{-l}$. 
For simplicity and illustration, Euler's method \cite{KlPl92} will be considered.
However, the results can easily be extended and the theory will be presented more generally.
In particular, 
\begin{eqnarray}
\label{eq:euler1step}
X^l_{k,(m+1)} & = & X^l_{k,m} + {h_l} a(X^l_{k,m}) + \sqrt{h_l} b(X^l_{k,m}) \xi_{k,m}, \\
\xi_{k,m} & \stackrel{\textrm{i.i.d.}}{\sim} & \mathcal{N}_d(0,I_d) 
\nonumber\end{eqnarray}
for $m=0,\dots, k_l$, where 
$k_l=2^l$ and $\mathcal{N}_d(0,I_d)$ is the $d-$dimensional normal distribution with mean zero and identity covariance 
(when $d=1$ 
the subscript is omitted). 
The numerical scheme gives rise to its own transition density between observation times 
$Q^l(x_{(k-1)},x )$, which is the density of $X^l_{(k-1),k_l}=X^l_{k,0}=X^l_{k}$, 
given initial condition $X^{l}_{(k-1),0}=x_{(k-1)}$.
Let $\heta_1^l(\varphi) := \bbE \varphi(X_1^l)$ for $l=0,\dots, \infty$.
Suppose one aims to approximate the expectation of 
$\varphi \in \cB_b(\bbR^d)$. 
For a given $L$, the
Monte Carlo approximation
of $\heta_1^\infty(\varphi)$ by 
$$
\heta_1^{L,N}(\varphi) = \frac{1}{N}\sum_{i=1}^N \varphi({X}^{L,i}_1), \qquad {X}^{L,i}_1 \sim Q^L(x_0,\cdot)\ ,
$$
has mean square error (MSE) given by 
 \begin{equation}
 \bbE | \heta_1^{L,N}(\varphi) -  \heta_1^\infty(\varphi) |^2 =   
  \underbrace{\bbE | \heta_1^{L,N}(\varphi) -  \heta_1^L(\varphi) |^2}_{\rm variance} +
    \underbrace{| \heta_1^{L}(\varphi) -  \heta_1^\infty(\varphi) |^2}_{\rm bias}.
 \label{eq:mse}
 \end{equation}
 If one aims for $\cO(\varepsilon^2)$ MSE with optimal cost, then one must balance these two terms.  

For $l=0,1,\ldots, L$, the hierarchy of time-steps $\{h_l\}_{l=0}^L$ 
gives rise to a hierarchy of transition densities $\{Q^l\}_{l=0}^L$.  In this context, 
for a single transition, it is well-known that the multilevel Monte Carlo (MLMC) method \cite{Giles08, heinrich2001multilevel} 
can reduce the cost to obtain a given level of mean-square error (MSE) \eqref{eq:mse}.
The description of this method and its extension to the particle filter setting will be the topic of the next 
section.

\section{Multilevel Particle Filters}
\label{sec:mlpf}

In this section, the multilevel particle filter will be introduced.  First, a review of the standard multilevel Monte-Carlo method is 
presented, 
illustrating the strategy for reducing the necessary cost for a given level of mean-square
error.  Next, the extension to the multilevel particle filter is presented.

\subsection{Multilevel Monte Carlo}
\label{ssec:mlmc}

The standard multilevel Monte Carlo (MLMC) framework \cite{Giles08} begins with asymptotic estimates for 
weak and strong error rates, and the associated cost.  
In particular, assume the following. 
\begin{asn}[MLMC Rates]  There are $\alpha, \beta, \gamma>0$ such that
\begin{itemize}
\item[(i)] 
$\bbE [ \varphi(X_{1}^{l}) - \varphi(X_{1}^{\infty}) ] = \cO(h_l^\alpha)$;
\item[(ii)] 
$\bbE [ | \varphi(X_{1}^{l}) - \varphi(X_{1}^{\infty}) |^p ]^{2/p} = \cO(h_l^\beta)$;
\item[(iii)] 
COST$(X_{1}^{l}) = \cO(h_l^{-\gamma})$,
\end{itemize}
where COST denotes the computational effort to obtain one sample $X_{1}^{l}$,   
and $h_l$ is the grid-size of the numerical method, 
for example the Euler method as given in \eqref{eq:euler1step}.  In this case 
$\alpha=\beta=\gamma=1$.
In general $\alpha\geq\beta/2$, as the choice $\alpha=\beta/2$ is always possible, by Jensen's inequality.
\label{asn:mlmcrates}
\end{asn}

Recall that in order to minimize the effort to obtain a given MSE, one must balance the terms in 
\eqref{eq:mse}.
Based on 
Assumption \ref{asn:mlmcrates}(i) above, a bias error proportional to $\varepsilon$ will require 
\begin{equation}
\label{eq:ell}
L \propto 
-\log(\varepsilon)/(\log(2)\alpha).
\end{equation}
The associated cost, 
in terms of $\varepsilon$, for a given sample is $\cO(\varepsilon^{-\gamma/\alpha})$.
Furthermore, the necessary number of samples to obtain a variance proportional to $\varepsilon^2$ for this 
standard single level estimator is given by $N \propto \varepsilon^{-2}$ following from \eqref{eq:mc}.
So the total cost to obtain a mean-square error tolerance of $\cO(\varepsilon^2)$ is: 
$\#$samples$\times($cost$/$sample)$=$total cost$\propto \varepsilon^{-2 - \gamma/\alpha}$.
To anchor to the particular example of the Euler-Marayuma method, 
the total cost is $\cO(\varepsilon^{-3})$.

Define a kernel 
$M^l: [\bbR^d\times \bbR^d] \times [\sigma(\bbR^d)\times \sigma(\bbR^d)] \rightarrow \bbR_+$, 
where $\sigma(\cdot)$ denotes the sigma algebra of measurable subsets,
such that $M_1^l (x,A) := M^l([x,x'],A\times\bbR^d) = Q^l(x,A)$
and $M_2^l (x',A) := M^l([x,x'],\bbR^d \times A) = Q^{l-1}(x',A)$.
The idea of MLMC is the following.  First approximate the $l^{th}$ increment 
$(\eta_1^l - \eta_1^{l-1})(\varphi)$ by an empirical average
\begin{equation}
Y_{l}^{N_l} (\varphi) := \frac{1}{N_l} \sum_{i=1}^{N_l} \varphi(X_{1,1}^{l,i}) - \varphi(X_{1,2}^{l,i}),
\label{eq:incemp}
\end{equation}
where $[X_{1,1}^{l,i},X_{1,2}^{l,i}] \sim M^l([x_0,x_0],\cdot)$, given initial datum $X_0=x_0$.
The multilevel estimator is a telescopic sum of such unbiased increment estimators, 
which yields an unbiased estimator of $\eta^L_1(\varphi)$.  It can be  
defined in terms of its empirical measure as
\begin{equation}
\heta_1^{L,{\rm Multi}} (\varphi) := \
\sum_{l=0}^L  Y_{l}^{N_l} (\varphi) \ , 
\end{equation}
under the convention that $\varphi(X^{0,i}_{1,2})\equiv 0$.

The mean-square error of the multilevel estimator is given by
\begin{equation}
\begin{split}
& \bbE \left \{ \heta_1^{L,{\rm Multi}} (\varphi) - \eta_1^\infty(\varphi) \right \}^2  =  \\
& \underbrace{ \sum_{\ell=0}^L \bbE \left \{ Y^{N_\ell}_{\ell}(\varphi) - [\eta_1^l(\varphi) - \eta_1^{l-1}(\varphi)] \right \}^2}_{\rm variance}  
+ \{ \underbrace{\eta_1^L(\varphi) - \eta_1^\infty(\varphi)}_{\rm bias} \}^2 .
\end{split}
\end{equation}
The key observation is that the bias is given by the {\it finest} level, 
whilst the variance is decomposed into a sum of variances of the {\it increments}
$\cV = \sum_{l=0}^L V_l N_l^{-1}$.
Sufficient correlation must be built into the kernels $M^l$ 
to ensure condition Assumption \ref{asn:mlmcrates}(ii)
above carries over to the increments
(for example two discretizations of the {\it same random realization} of the SDE \eqref{eq:sde}).
Then the variance of the $l^{th}$ increment has the form $V_l N_l^{-1}$ and $V_l = \cO(h_l^\beta)$ following from Assumption \ref{asn:mlmcrates} (ii), allowing smaller number of samples $N_l$ at cost $C_l = \cO(h_l^{-\gamma})$
 for larger $l$, following from Assumption \ref{asn:mlmcrates}(iii). 
 The total cost is given by the sum $\cC=\sum_{l=0}^L C_l N_l$.
Based on Assumption \ref{asn:mlmcrates}(ii) and Assumption \ref{asn:mlmcrates}(iii) 
above, 
optimizing $\cC$ for a fixed $\cV$ yields that $N_l = \lambda^{-1/2} 2^{-(\beta+\gamma)l/2}$, 
for Lagrange multiplier $\lambda$. 
In the Euler-Marayuma case $N_l = \lambda^{-1/2} 2^{-l}$.  
Now, one can see that after fixing the bias to $c\varepsilon$,
one aims to find the 
Lagrange multiplier $\lambda$ such that $\cV 
\approx c^2\varepsilon^2$.  
Defining $N_0=\lambda^{-1/2}$, then $\cV 
= N_0^{-1} \sum_{l=0}^L 2^{(\gamma-\beta)l/2}$, 
so one must have $N_0 \propto \varepsilon^{-2} K(\varepsilon)$, where $K(\varepsilon)=\sum_{l=0}^L 2^{(\gamma-\beta)l/2}$,
and the $\varepsilon$-dependence comes from $L(\varepsilon)$, as defined in \eqref{eq:ell}.
There are three cases, with associated $K$, and hence cost $\cC$, given in Table \ref{tab:mlcases}.

\begin{table}[h]
\begin{center}
  \begin{tabular}{ | c || c | c |}
    \hline
    CASE & $K(\varepsilon)$ & $\cC(\varepsilon)$ \\ \hline\hline
    $\beta>\gamma$ & $\cO(1)$ & $\cO(\varepsilon^{-2})$ \\ \hline
    $\beta=\gamma$  & $\cO(-\log(\varepsilon))$ & $\cO(\varepsilon^{-2}\log(\varepsilon)^2)$ \\ \hline
    $\beta<\gamma$  & $\cO(\varepsilon^{(\beta-\gamma)/(2\alpha)})$ & $\cO(\varepsilon^{-2+(\beta-\gamma)/\alpha})$ \\
    \hline
  \end{tabular}
\end{center}
\caption{The three cases of multilevel Monte Carlo, and associated constant $K(\varepsilon)$ and cost $\cC(\varepsilon)$.}
\label{tab:mlcases}
\end{table}

For example, Euler-Marayuma falls into the case 
($\beta=\gamma$), so that $\cC(\varepsilon)= \cO(\varepsilon^{-2}\log(\varepsilon)^2)$.  In this case, one chooses 
$N_0 = C \varepsilon^{-2} |\log(\varepsilon)| = C 2^{2L} L$, where the purpose of $C$ is 
to match the variance with the bias$^2$, similar to the single level case.

 The kernel $M^l$ can be constructed using the following strategy.
First the finer discretization is simulated using \eqref{eq:euler1step} (ignoring index $k$) with
$X^{l,i}_{0,1}=x_0$, 
for $i\in\{1,\dots, N_l\}$.
Now for the coarse discretization, let 
$X^{l,i}_{0,2}=x_0$ for $i\in\{1,\dots,N_l\}$, let $h_{l-1}=2h_l$ and 
for $m\in\{1,\dots,k_{l-1}\}$ simulate
\begin{equation}
X^{l,i}_{m+1,2} =  X^{l,i}_{m,2} + h_{l-1} a(X^{l,i}_{m,2}) + \sqrt{h_{l-1}} b(X^{l,i}_{m,2}) (\xi^i_{2m}+\xi^i_{2m+1}),  
\label{eq:euler1stepcoarse}
\end{equation}
where $\{\xi^i_{m}\}_{i=1,m=0}^{N_l,k_l}$ are the $i^{th}$ realizations used in the 
simulation of the finer discretization.
This procedure defines a kernel $M^l$ as above, such that 
$(X^{l,i}_{k_{l-1},1},X^{l,i}_{k_{l-1},2}) \sim M^l([x_0,x_0], ~\cdot~)$ are suitably coupled 
and the standard MLMC theory will go through 
with $\alpha=\beta=\gamma=1$ above.

\subsection{Multilevel Particle Filters}
\label{ssec:mlpf}

The framework of the previous section will now be extended to the new
multilevel particle filter (MLPF). Throughout, the observations $y_{1:m}$ are omitted from the notations.
It will be convenient to define $U^l_m:=X^l_{m}|y_{1:m-1}$ for $l=0,\ldots,\infty$, with $U^\infty_m:=X^\infty_{m}|y_{1:m-1}$
denoting the limiting continuous-time process, 
and denote the associated predicting distributions by $\eta_m^l$.
It will also be useful to define $\widehat{U}^l_m:=X^l_{m}|y_{1:m}$, and its distribution $\heta_m^l$.
Let $\varphi \in \cB_b(\bbR^d)$ and consider the following decomposition
\begin{eqnarray}
\eta^\infty_m(\varphi) & =  & \sum_{l=0}^L (\eta^{l}_m - \eta^{l-1}_m)(\varphi)  +  (\eta^\infty_m - \eta^{L}_m)(\varphi)
\label{eq:collapse_sum}
\end{eqnarray}
where 
$\eta_m^{-1}(\varphi) :=0$. 

Let $U^{l,i}_{0,1}=\widehat{U}^{l,i}_{0,1} = U^{l,i}_{0,2} = \widehat{U}^{l,i}_{0,2} = X_0^i$, where 
$X_0^i \sim \eta_0 = \heta_0$, and iterate the following.  
Draw $[U^{l,i}_{m,1},U^{l,i}_{m,2}] \sim M^l ([\widehat{U}^{l,i}_{m-1,1},\widehat{U}^{l,i}_{m-1,2}],~\cdot~)$
Each summand in the first term of \eqref{eq:collapse_sum} can be estimated with:
$$
\sum_{i=1}^{N_l}\left \{ w_{m,1}^{l,i} \varphi(U^{l,i}_{m,1}) - 
w_{m,2}^{l,i} \varphi(U^{l,i}_{m,2}) \right \},
$$
where the weights are defined as follows, for $h\in\{1,2\}$, 
\begin{equation}
w_{m,h}^{l,i} = \frac{G(y_{m},U^{l,i}_{m,h})}{\sum_{j=1}^{N_l} G(y_{m},U^{l,j}_{m,h})}. 
\label{eq:weights}
\end{equation}
It is clear that for suitably well-behaved $G$, for example satisfying Assumption \ref{asn:g}, such an estimate will satisfy the standard MLMC identity and cost.
 However, it is well-known that one must perform resampling in order for a particle filter to perform well for multiple steps.
 Here this is a particularly challenging point, as the samples have to remain suitably coupled after the resampling, so that similar rates hold as above.

For every index $k\in\{1,\dots,N_l\}$ 
the indices $I^{l,k}_{m,j}$, $j\in\{1,2\}$, are sampled according to the {\bf coupled resampling} 
procedure described below: 

\begin{itemize}
\item[{\bf a}.] with probability  $\alpha_m^l = \sum_{i=1}^{N_l}w_{m,1}^{l,i}\wedge w_{m,2}^{l,i}$, 
draw $I^{l,k}_{m,1}$ according to
$$
\bbP(I^{l}_{m,1}=i) = \frac{1}{\alpha_m^l} (w_{m,1}^{l,i}\wedge w_{m,2}^{l,i}),
\qquad i=1,\ldots,N_l.
$$
and let $I^{l,k}_{m,2}=I^{l,k}_{m,1}$.
\item[{\bf b}.] Define $Z^l_{m,h} := w_{m,h}^{l,i}-w_{m,1}^{l,i}\wedge w_{m,2}^{l,i}$, and 
with probability $1-\alpha_m^l$, draw
$(I^{l,k}_{m,1}, I^{l,k}_{m,2})$ 
{\it independently} according to the probabilities 
\[
\begin{split}
\bbP(I^{l}_{m,1}=i) & =  Z^{l,i}_{m,1} / \sum_{j=1}^{N_l} Z^{l,j}_{m,1} \, ; \\
\bbP(I^{l}_{m,2}=i) & =  Z^{l,i}_{m,2} / \sum_{j=1}^{N_l} Z^{l,j}_{m,2} \, ,
\end{split}
\]
for $i=1,\ldots,N_l$.
\end{itemize}

The indices for the fine (resp.~coarse) discretization 
are resampled marginally according to 
$w_{1}^{l,i}$ (resp.~$w_{2}^{l,i}$), which is exactly as required.  
Notice that it is necessary to independently sample the fine and coarse levels with a small probability
in order to preserve the marginals.  However, it will be shown that 
the resulting samples do remain sufficiently coupled, 
although with a slightly lower rate than the vanilla MLMC.
Finally the {\bf multilevel particle filter (MLPF)} is given below: 

{\bf For} $l=0,1,\dots,L$ and $i=1,\dots, N_l$, draw $\hU^{l,i}_{0,1}\sim \mu_0$, and let 

$\hU^{l,i}_{0,2}=\hU^{l,i}_{0,1}$.

{\bf Initialize} $m=1$.  {\bf Do} 

\begin{itemize}
\item[(i)] {\bf For} $l=0,1,\dots,L$ and $i=1,\dots, N_l$, draw $(U^{l,i}_{m,1},U^{l,i}_{m,2}) \sim M^l((\hU^{l,i}_{m-1,1},\hU^{l,i}_{m-1,2}),~\cdot~)$;
\item[(ii)] {\bf For} $l=0,1,\dots,L$ and $k=1,\dots, N_l$, draw $(I^{l,k}_1,I^{l,k}_2)$ according to the coupled resampling procedure above;
\item[(iii)] $(\hU^{l,k}_{m,1},\hU^{l,k}_{m,2}) \leftarrow (U^{l,I^{l,k}_1}_{m,1},U^{l,I^{l,k}_2}_{m,2})$.
\end{itemize}

$m \leftarrow m+1$
\newline

Note that if the variance of the weights becomes substantial, one can use the approach in \cite{jasra} to 
deal with this issue.

\section{Theoretical Results}
\label{sec:theo}

The calculations leading to the results in this section are performed via a Feynman-Kac type representation 
(see \cite{delm:04,delmoral1}) which is detailed in the supplementary material.  
Denote the marginal transition kernels of the Euler discretization procedure described above at level $l$ 
as $M^l_{1}$ (fine) and $M^l_{2}$ (coarse). 
Note that these results do not depend on Euler discretization and hold for any general coupled particle filter. 
Also note that the results are easily extended to non-autonomous SDE \eqref{eq:sde}, at the expense of additional technicalities.  
The predictor at time $m$, level $l$, is denoted as $\eta^l_{m,1}$ (fine) and $\eta^l_{m,2}$ (coarse). 
$\mathcal{B}_b(\mathbb{R}^d)$ are the bounded, measurable and real-valued functions
on $\mathbb{R}^d$ and $\textrm{Lip}(\mathbb{R}^d)$ are the globally Lipschitz real-valued functions on $\mathbb{R}^d$. 
Denote the supremum norm as $\|\cdot\|$, and the 
total variation norm 
as $\|\cdot\|_{\textrm{tv}}$.
For two Markov kernels $M_1$ and $M_2$
on the same space $E$, letting $\cA=\{\varphi: \|\varphi\|\leq 1, \varphi\in\textrm{Lip}(E)\}$
write
$$
|||M_{1}-M_{2}||| := \sup_{\varphi\in \cA}\sup_x |\int_E \varphi(y) M_1(x,dy) - \int_E \varphi(y) M_2(x,dy) |.
$$
Let $w_{m,j}^{l,i}$ denote the weights defined as in \eqref{eq:weights}
with the index $m$ indicated explicitly.
For each $j\in\{1,2\}$, $p\geq 1$, $m\geq 1$ define
\[\begin{split}
& M_{m,j}(u_p,du_{p+m}) = \\
& \int_{\bbR^{d\times m-1}} M_{j}(u_p,du_{p+1}) 
\cdots M_{j}(u_{p+m-1},du_{p+m}).
\end{split}\]
Finally, the following notation is introduced for the selection densities
$G_m(\cdot) := G(y_{m+1},\cdot)$.

The following assumption will be made, uniformly over the level $l \in [0,1,\ldots)$,
which will be omitted for notational simplicity. 
\begin{asn}[Mutation] 
There exists a 
$C>0$ such that for each $u,u'\in \bbR^d$, $j\in\{1,2\}$ and $\varphi\in\mathcal{B}_b(\bbR^d)\cap\textrm{Lip}(\bbR^d)$
$$
|M_{j}(\varphi)(u) - M_{j}(\varphi)(u')| \leq C\|\varphi\|~|u-u'|.
$$
\label{asn:liperg}
\end{asn}

Additionally, it will be assumed that for all suitable test-functions 
$\varphi\in\mathcal{B}_b(\cU)\cap\textrm{Lip}(\cU)$ the following hold.
\begin{asn}[MLPF rates] 
For $l \in [0,1,\ldots)$, and $p\geq 1$, let $(U^{l}_1, U^{l}_2) \sim M^l((U_{0,1}^l,U_{0,2}^{l}),~\cdot~)$, 
where $\bbE[\varphi(U_{0,1}^{l}) - \varphi(U_{0,2}^{l})]= \cO(h_l^\alpha)$ and
 $\bbE [ | \varphi(U_{0,1}^{l}) - \varphi(U_{0,2}^{l}) |^p ]^{2/p} = \cO(h_l^\beta)$
 for some $\alpha\geq\beta/2>0$.
 Then, there is a $\gamma>0$ such that
\begin{itemize}
\item[{\rm (i)}] 
{\rm max}$\left\{\left|\bbE[\varphi(U_1^{l}) - \varphi(U^{l}_2)]\right|, 
|||M^l_{1}-M^l_{2}||| \right \} = \cO(h_l^\alpha)$;
\item[{\rm (ii)}] 
$\bbE [ | \varphi(U^{l}_1) - \varphi(U^{l}_2) |^p ]^{2/p} = \cO(h_l^\beta)$;
\item[{\rm (iii)}] 
COST$[M^l] 
= \cO(h_l^{-\gamma})$,
\end{itemize}
where COST$[M^l]$ 
is the cost to simulate one sample from the kernel $M^l$.
\label{asn:mlrates}
\end{asn}

\subsection{Main Result}
\label{ssec:mlpf}

Here the MLPF theorem is presented, followed by the main theorem upon which it is based.  The proof and supporting lemmas are provided in the supplementary materials.
Let 
\begin{equation}
A^{N_l}_{l,m}(\varphi)=\sum_{i=1}^{N_l}[ w_{m,1}^{l,i}\varphi(U^{l,i}_{m,1}) - w_{m,2}^{l,i}\varphi(U^{l,i}_{m,2}) ], 
\label{eq:increst}
\end{equation} 
with the convention that $w_{m,2}^{0,i}:=0$, and define
$\heta^{\rm ML}_m(\cdot) := \sum_{l=0}^LA^{N_l}_{l,m}(\cdot)$.

\begin{theorem}[MLPF]
\label{thm:mlpf}
Let Assumptions \ref{asn:g}, \ref{asn:liperg}, and \ref{asn:mlrates} be given. 
Then for any $m\geq 0$, 
$\varphi\in \mathcal{B}_b(\mathbb{R}^d)\cap\textrm{\emph{Lip}}(\mathbb{R}^d)$, 
and $\varepsilon>0$,
there exists a finite constant $C(m,\varphi)$, an $L>0$, and $\{N_l\}_{l=0}^L$ 
such that 
\[
\mathbb{E}\Bigg[
\Bigg(\heta^{\rm ML}_m(\varphi)
- \heta_m^\infty(\varphi) \Bigg)^2   \Bigg] \leq
C(m,\varphi) \varepsilon^2, 
\]
for the cost $\cC(\varepsilon)$ given in the third column of Table \ref{tab:mlpfcases}.

\begin{table}[h]
\begin{center}
  \begin{tabular}{ | c || c | c |}
    \hline
    CASE & $K(\varepsilon)$ & $\cC(\varepsilon)$ \\ \hline\hline
    $\beta>2\gamma$ & $\cO(1)$ & $\cO(\varepsilon^{-2})$ \\ \hline
    $\beta=2\gamma$  & $\cO(-\log(\varepsilon))$ & $\cO(\varepsilon^{-2}\log(\varepsilon)^2)$ \\ \hline
    $\beta<2\gamma$  & $\cO(\varepsilon^{(\beta-2\gamma)/(4\alpha)})$ & $\cO(\varepsilon^{-2+(\beta-2\gamma)/(2\alpha)})$ \\
    \hline
  \end{tabular}
\end{center}
\caption{The three cases of MLPF, and associated constant $K(\varepsilon)$ and cost $\cC(\varepsilon)$.}
\label{tab:mlpfcases}
\end{table}

\end{theorem}

{\it Proof. }
Notice that 
\[
\begin{split}
\mathbb{E}\Big[
\Big(
\heta_m^{\rm ML}(\varphi)
- \heta^\infty_m(\varphi) \Big)^2   \Big]  & \leq 
  2\mathbb{E}\Big[
\Big(
\heta_m^{\rm ML}(\varphi) - \heta_m^L(\varphi) \Big)^2   \Big] \\
& + 2 \Big( \heta_m^L(\varphi)
- \heta_m^\infty(\varphi) \Big)^2.
\end{split}\]
First, note that a theoretical kernel $M^{L,\infty}$
can be defined to generate coupled pairs of particles $(U_{m,1}^{L,\infty}, U^{L,\infty}_{m,2})$ 
for $m\geq 1$ with marginals 
$U^{L,\infty}_{m,1} \sim \heta^\infty_m$ and $U^{L,\infty}_{m,2} \sim \heta^L_m$ satisfying
the Assumptions \ref{asn:mlrates}.  Assumption \ref{asn:g}(i) then ensures the rate carries over to the update
and finally induction shows the second term is $\cO(h_l^{2\alpha})$.
The rest of the proof follows from 
Theorems \ref{thm:mlpf1} and \ref{thm:main}, and Corollary \ref{cor:finrate}, 
noting that the terms in Corollary \ref{cor:finrate} are analogous to the $V_l$ 
terms from the standard multilevel theory described in the previous section.
Therefore, upon choosing $L\propto -\log(\varepsilon)$, and $N_l \propto N_0 2^{-(\beta+2\gamma)l/4}$ with 
$N_0 \propto \varepsilon^{-2} K(\varepsilon)$ and $K(\varepsilon)$ as in the second column of 
Table \ref{tab:mlpfcases}, the results follow exactly as for MLMC above. 
$\square$

This Theorem can be immediately applied to the particular example of the diffusion 
\eqref{eq:sde}, with appropriate discretization method.  
This is made explicit and precise in the following Corollary.

\begin{cor} Theorem \ref{thm:mlpf} holds for the diffusion example \eqref{eq:sde} 
under Assumptions \ref{asn:diff}, given a numerical method which satisfies Assumptions \ref{asn:mlrates}.  
Furthermore Assumptions \ref{asn:mlrates} hold 
for Euler-Marayuma method, with $\alpha=\beta=\gamma=1$.  
For a constant diffusion $b(x)=b$, one has $\beta=2$.
\label{cor:euler_mlpf}
\end{cor}

{\it Proof. }
Assumptions \ref{asn:diff} on \eqref{eq:sde} guarantee the required Assumptions \ref{asn:liperg} 
on the kernels $M^{L,\infty}$ \cite{oksendal1995}.  
For Euler-Marayuma method the kernels $M^l$ also satisfy Assumptions \ref{asn:liperg} and
\ref{asn:mlrates} \cite{GrahamTalay, MR1843179}, 
and the rates can be found in \cite{GrahamTalay, KlPl92}.  The improved rate $\beta=2$ for $b(x)=b$ 
is well-known, as the Euler method coincides with the Milstein method in the case of constant diffusion 
\cite{GrahamTalay}.
$\square$

The main theorem which provides the appropriate convergence rate for the MLPF Theorem \ref{thm:mlpf} is now presented.

\begin{theorem}
Assume \ref{asn:liperg} for each level for the mutation kernel(s) and \ref{asn:g} for the updates.
Then for any $m\geq 0$, $1\leq L<+\infty$, $\varphi\in \mathcal{B}_b(\mathbb{R}^d)\cap\textrm{\emph{Lip}}(\mathbb{R}^d)$, there exists a
constant $C(m,\varphi) = \max_{0\leq l \leq L} C_l(m,\varphi)$  
such that
\[
\begin{split}
& \mathbb{E}\Bigg[
\Bigg(
\heta_m^{\rm ML}
- \heta_m^L(\varphi) \Bigg)^2   \Bigg] \leq \\
& C(m,\varphi) \sum_{l=0}^L \frac{1}{N_l}
\Bigg( B_l(m) + \sum_{q\neq l=0}^L \frac{\sqrt{B_l(m)B_q(m)}}{N_q}
\Bigg)
\end{split}
\]
\begin{align}
\nonumber
B_l(n)  =& \Big(\sum_{p=0}^n \mathbb{E}[\{|u_{p,1}^{l,1}-u_{p,2}^{l,1}|\wedge 1 \}
^{2}]^{1/2}+\|\eta_{p,1}^l-\eta_{p,2}^l\|_{\textrm{tv}} \\
& + \sum_{p=1}^n|||M_{p,1}^l-M_{p,2}^l|||\Big)^2.
\label{eq:bee}
\end{align}
\noindent 
Subscripts are added to indicate level-dependence, and the constants have been absorbed into the single one.  
\label{thm:mlpf1}
\end{theorem}

{\it Proof. }
Let $\tilde{A}_{l,m}^{N_l} (\cdot) = \Big(A_{l,m}^{N_l} - (\heta_m^l - \heta_m^{l-1})\Big)(\cdot),$ where $A_{l,m}^{N_l}$ 
is defined in Equation \eqref{eq:increst}, with $\heta_m^{-1} := 0$.
Noting the independence {\it between} increments, the telescoping sum provides 
\[
\begin{split}
\mathbb{E}\Bigg[
\Big(\sum_{l=0}^L\tilde{A}_{l,m}^{N_l}(\varphi) \Big)^2   \Bigg] & =  
\sum_{l=0}^L \Bigg( \mathbb{E} \Big[\Big(\tilde{A}_{l,m}^{N_l}(\varphi)\Big)^2 \Big] \\
& +  \sum_{q\neq l=0}^L \mathbb{E} \Big(\tilde{A}_{l,m}^{N_l}(\varphi)\Big) 
\mathbb{E} \Big(\tilde{A}_{q,m}^{N_q}(\varphi)\Big) \Bigg).  
\end{split}
\]  
The bound therefore follows trivially from applying 
Theorems \ref{theo:filt_error} 
and Lemma \ref{lem:inc_bias}
from the Supplementary materials 
to each level.  
$\square$

The bound of the first term in $B_l(n)$ of \eqref{eq:bee} is limited by 
the coupled resampling, 
and is asymptotically proportional to $h_l^{\beta/2}$.  This is the reason for the reduced rate.

\section{Numerical Examples}

\subsection{Model Settings}

The numerical performance of the MLPF algorithm will be illustrated here, with a few
examples of the diffusion processes considered in this paper. Recall that the
diffusions take the following form
\begin{equation*}
  d X_t = a(X_t) dt + b(X_t) d W_t, \qquad X_0 = x_0
\end{equation*}
with $X_t\in\bbR^d$, $t\ge0$ and $\{W_t\}_{t\in[0,T]}$ a Brownian motion of
appropriate dimension. In addition, partial observations $\{y_1,\dots,y_n\}$
are available with $Y_k$ obtained at time $k\delta$, and $Y_k|X_{k\delta}$ has
a density function $G(y_k,x_{k\delta})$. 
The objective is the estimation of
$\Exp[\varphi(X_{k\delta})|y_{1:n}]$ for some test function $\varphi(x)$.
Details of each example are described below. A summary of settings can be found in
Table~\ref{tab:model}.

\paragraph{Ornstein-Uhlenbeck Process}

First, consider the following \ou process,
\begin{gather*}
  d X_t = \theta(\mu - X_t)dt + \sigma d W_t, \\
  Y_k|X_{k\delta} \sim \calN(X_{k\delta}, \tau^2), \qquad \varphi(x) = x.
\end{gather*}
An analytical solution exists for this process and the
exact value of $\Exp[X_{k\delta}|y_{1:k}]$ can be computed using a Kalman
filter. The constants in the example are, $x_0 = 0$, $\delta = 0.5$, $\theta =
1$, $\mu = 0$, $\sigma = 0.5$, and $\tau^2 = 0.2$.

\paragraph{Geometric Brownian Motion}

Next consider the \gbm process,
\begin{gather*}
  d X_t = \mu X_t dt+ \sigma X_t d W_t, \\
  Y_k|X_{k\delta} \sim \calN(\log X_{k\delta}, \tau^2), \qquad
  \varphi(x) = x,
\end{gather*}
This process 
also admits an analytical solution, by using the transformation $Z_t = \log X_t$.
The constants are, $x_0 = 1$, $\delta = 0.001$, $\mu = 0.02$, $\sigma = 0.2$
and $\tau^2 = 0.01$.

\paragraph{Langevin Stochastic Differential Equation}

Here the \sde is given by 
\begin{gather*}
  d X_t = \frac{1}{2}\nabla\log\pi(X_t) dt+ \sigma dW_t, \\
  Y_k|X_{k\delta} \sim \calN(0, \tau^2e^{X_{k\delta}}),\qquad
  \varphi(x) = \tau^2e^x
\end{gather*}
where $\pi(x)$ denotes a probability density function. 
The density $\pi(x)$ is chosen as the Student's $t$-distribution with degrees of freedom $\nu = 10$. 
The other constants are, $x_0 = 0$, $\delta = 1$, $\sigma = 1$ and $\tau^2 = 1$.
Real daily \sp log return data (from August 3, 2011 to July 24, 2015,
normalized to unity variance) is used.

\paragraph{An \sde with a Non-Linear Diffusion Term}

Last, 
the following \sde is considered,
\begin{gather*}
  d X_t = \theta(\mu - X_t)dt + \frac{\sigma}{\sqrt{1 + X_t^2}} d W_t, \\
  Y_k|X_{k\delta} \sim \calL(X_{k\delta}, s), \qquad \varphi(x) = x,
\end{gather*}
where $\calL(m,s)$ denotes the Laplace distribution with location $m$ and scale
$s$. The constants are $x_0 = 0$, $\delta = 0.5$, $\theta = 1$, $\mu = 0$,
$\sigma = 1$ and $s = \sqrt{0.1}$. 
This example is abbreviated $\nlm$ 
in the remainder of this section. 

\begin{table}
  \begin{tabu}{X[2l]X[2l]X[l]X[2l]X[l]}
    \toprule
    Example & $a(x)$ & $b(x)$ & $G(y;x)$ & $\varphi(x)$ \\
    \midrule
    \ou
    & $\theta(\mu - x)$
    & $\sigma$
    & $\calN(x, \tau^2)$
    & $x$ \\
    \gbm
    & $\mu x$
    & $\sigma x$
    & $\calN(\log x, \tau^2)$~~~~
    & $x$ \\
    Langevin
    & $\frac{1}{2}\nabla\log\pi(x)$~~~
    & $\sigma$
    & $\calN(0, \tau^2e^x)$
    & $\tau^2 e^x$ \\
    \nlm
    & $\theta(\mu - x)$
    & $\frac{\sigma}{\sqrt{1 + x^2}}$
    & $\calL(x, s)$
    & $x$ \\
    \bottomrule
  \end{tabu}
  \caption{Model settings}
  \label{tab:model}
\end{table}

\subsection{Simulation Settings}

For each example, 
multilevel estimators are considered at levels $L = 1,\dots,8$. For the \ou
and \gbm processes, the ground truth is computed through a Kalman filter. For
the two other examples, 
results from particle filters at level $L = 9$ are used 
as approximations to the ground truth.

For each level of \mlpf algorithm, $N_l = \lfloor N_{0,L} h_l^{(\beta
  + 2 \gamma) / 4} \rfloor$ particles are used, where $h_l = M_l^{-1} = 2^{-l}$
is the width of the Euler-Maruyama discretization; $\gamma$ is the rate of
computational cost, which is $1$ for the examples considered here; and $\beta$
is the rate of the strong error. The 
value of $\beta$ is $2$ if the
diffusion term $b(x)$ is constant and $1$ in general. 
The value $N_{0,L} \propto \varepsilon^{-2} K(\varepsilon)$ 
is set according to Table \ref{tab:mlpfcases}.  
For the cases in which the diffusion term is constant, we let $N_{0,L} = 2^{2L} L$,
while for the other cases $N_{0,L} = 2^{(9/4)L} $.
Resampling is done adaptively. For the plain particle filters, resampling is
done when \ess is less than a quarter of the particle numbers. For the coupled
filters, we use the \ess of the coarse filter as the measurement of
discrepancy. Each simulation is repeated 100 times.

\subsection{Results}

First consider the rate $\beta/2$ of the strong error. This rate can be
estimated either by the sample variance of $\hat\varphi_l(X_{n\delta}) =
\sum_{i=1}^{N_l}\{w_1^{l,i}\varphi(U_{n,1}^{l,i}) -
w_2^{l,i}\varphi(U_{n,2}^{l,i})\}$, or by $1 - p_l(n)$, where $p_l(n)$ is the
probability of the coupled particles having the same resampling index at time
step $n$. Both $\var[\hat\varphi_l(X_{n\delta})]$ and $p_l(n)$ can be estimated
using the samples from \mlpf simulations. Figures~\ref{fig:rate_v}
and~\ref{fig:rate_p} show the estimated variance and value of $1 - p_l(n)$
against $h_l$, respectively. The estimated rates for the \ou and Langevin
examples are about $1$. 
For the other two examples, where
the diffusion term $b(x)$ is non-constant, the estimated rates are about $0.5$.
This is consistent with Corollary \ref{cor:euler_mlpf}.

\begin{figure}
  \includegraphics[height=1\columnwidth,width=1\columnwidth]{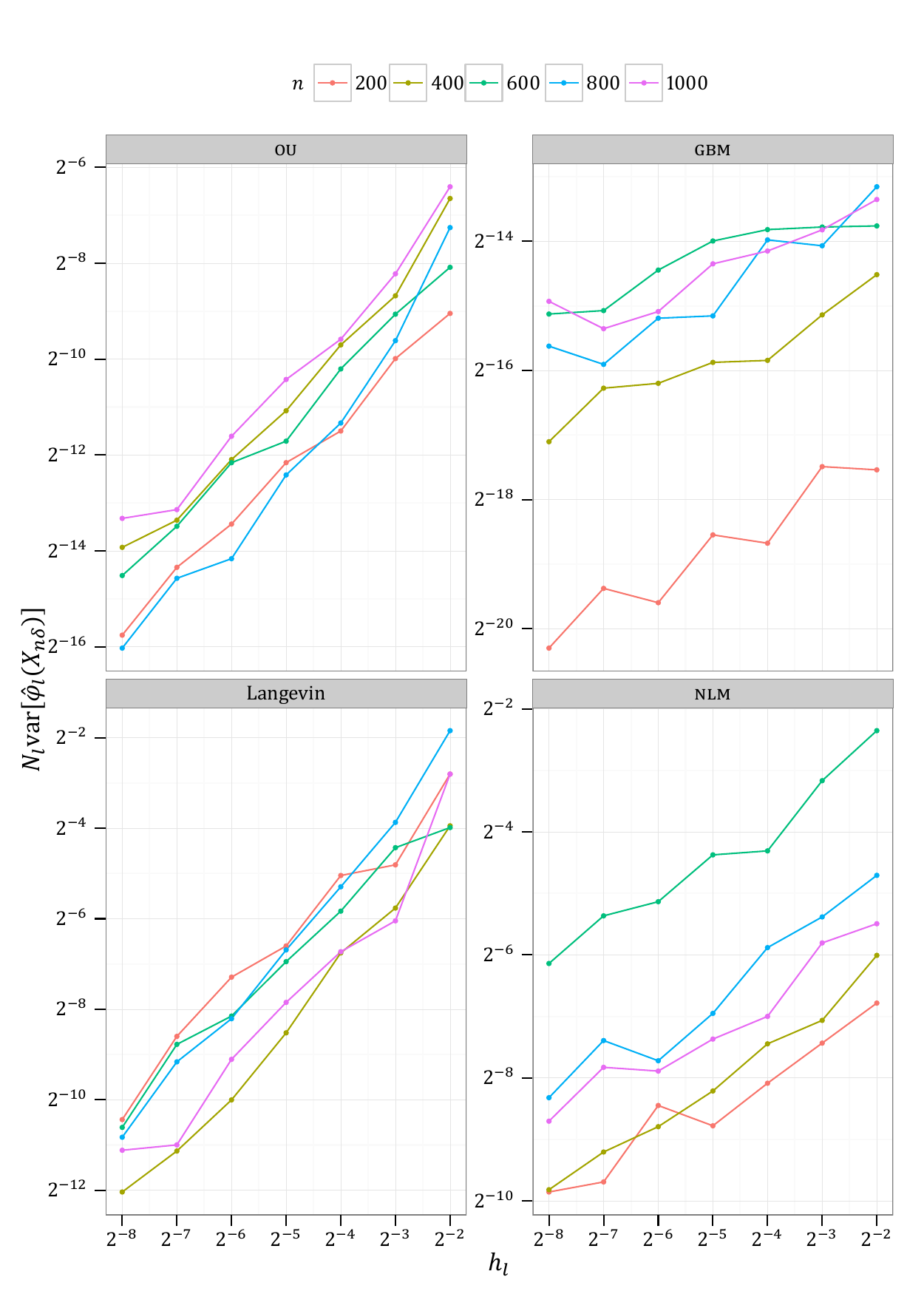}
  \caption{Rate estimates using the variance.}
  \label{fig:rate_v}
\end{figure}

\begin{figure}
  \includegraphics[height=1\columnwidth,width=1\columnwidth]{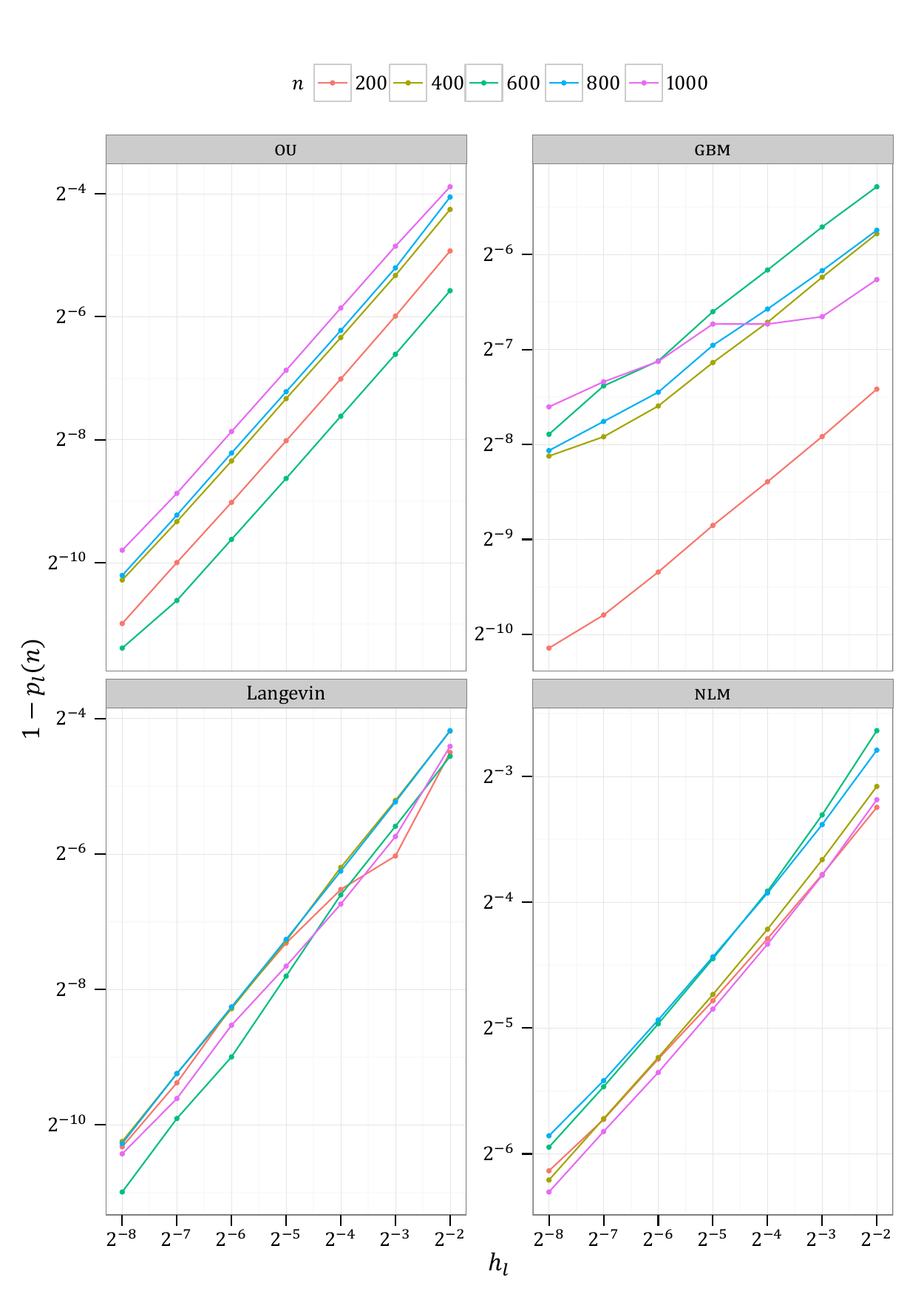}
  \caption{Rate estimates using the probability of coupling.}
  \label{fig:rate_p}
\end{figure}

Next 
the rate of cost 
vs. \mse is examined. 
This is shown
in Figure~\ref{fig:cost} and Table~\ref{tab:cost} for the estimator of
$\Exp[\varphi(X_{n\delta})|y_{1:n}]$.  This agrees with the theory, which
predicts a rate of $-1.5$ for the particle filter and a rate of $-1.25$ for the 
non-constant diffusion cases, and a logarithmic penalty on $-1$ for the others.

\begin{figure}
 \includegraphics[height=1\columnwidth,width=1 \columnwidth]{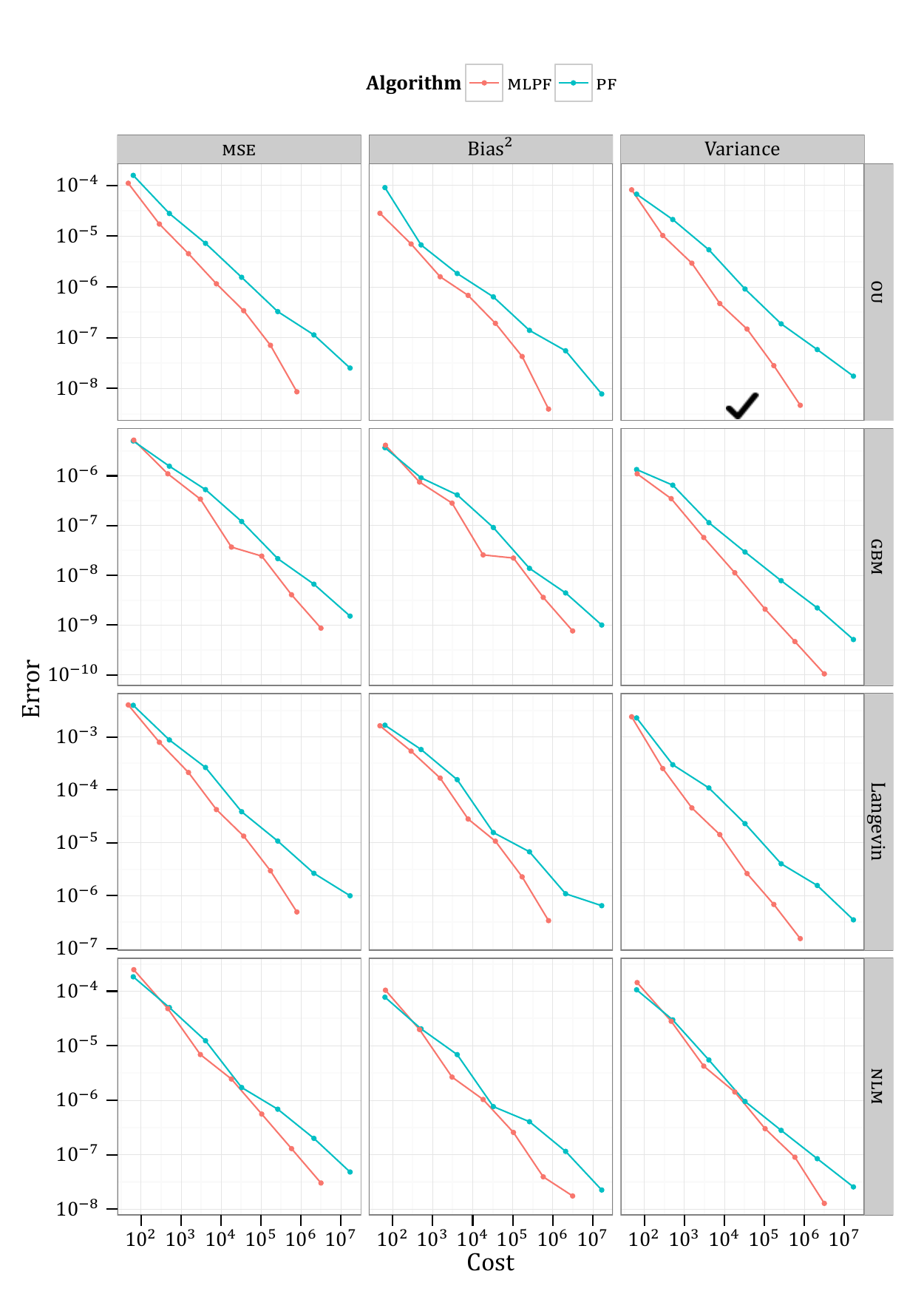}
  \caption{Cost rates as a function of MSE.}
  \label{fig:cost}
\end{figure}

\begin{table}
  \begin{tabu}{X[l]X[r]X[r]}
    \toprule
    Example & \pf & \mlpf \\
    \midrule
    \ou      & $-1.44$ & $-1.07$ \\
    \gbm     & $-1.51$ & $-1.24$\\
    Langevin & $-1.46$ & $-1.10$ \\
    \nlm     & $-1.50$ & $-1.21$ \\
    \bottomrule
  \end{tabu}
  \caption{Cost rate $\log \cC 
  \sim \log\mse$.}
  \label{tab:cost}
\end{table}

\section{Conclusions}

In this article 
a multilevel version of the particle filter has been introduced. The improvements that may be brought about by this 
approach were illustrated both theoretically and numerically. 
There are several natural extensions to this work.
First, and perhaps most importantly, is to theoretically understand the advantage of the particular
coupled resampling mechanism 
adopted in this article, in comparison to other types of coupled resampling,   
e.g.~via the variance in the CLT.  It is remarked that 
other resampling strategies were tried on these examples,
and they did not preserve a desired rate of strong convergence. However empirical results recently appeared in 
\cite{reich} which indicate that more favorable convergence rates may be preserved in certain cases 
by replacing the resampling step with a deterministic transformation. 
Second, it would be of interest to explore techniques for improving the preservation 
of coupling such that the same MLMC rate $\beta$ carries through to the MLPF, rather than $\beta/2$, 
e.g.~via coupling the independent pairs of particle filters in some way, 
or perhaps through a different resampling strategy involving antithetic variables \cite{giles_acta}.
Finally, one can use the approach in e.g.~\cite{jasra} to improve the stability of the particle filtering algorithm.

\subsubsection*{Acknowledgements}
AJ, KL \& YZ were supported by an AcRF tier 2 grant: R-155-000-143-112. AJ is affiliated with the Risk Management Institute and the Center for Quantitative Finance at NUS.  KL was supported by a Laboratory Directed Research and Development (LDRD) Strategic Hire grant from Oak Ridge National Laboratory (ORNL).  KL \& AJ were additionally supported by King Abdullah University of Science and Technology (KAUST).  KK was supported by CREST, JST and Grant-in-Aid for Young Scientists (B) 24740062.

\appendix

\section{Set Up}

\subsection{Basic Notations}

Consider a sequence of random variables $(v_n)_{n\geq 0}$ with $v_n=(u_{n,1},u_{n,2})\in \cU\times \cU =: \cV$. 
For $\mu\in\mathcal{P}(\cV)$ (the probability measures on $\cV$) and function $\varphi\in\mathcal{B}_b(\cU)$ (bounded-measurable, real-valued) we will write:
$$
\mu(\varphi_j) = \int_\cV \varphi(u_j) \mu(dv)\qquad j\in\{1,2\},\ v=(u_1,u_2).
$$
Write the $j\in\{1,2\}$ marginals (on $u_j$) of a probability $\mu\in\mathcal{P}(\cV)$ as $\mu_j$.
Define the potentials:
$G_n:\cU\rightarrow\mathbb{R}_+$.\ 
Let $\eta_0\in\mathcal{P}(\cV)$ and define
Markov kernels $M_{n}:\cV\rightarrow \mathcal{P}(\cV)$
and $M_{n,j}:\cU\rightarrow \mathcal{P}(\cU)$ with $n\geq 1$ and $j\in\{1,2\}$. It is explictly assumed that for $\varphi\in\mathcal{B}_b(\cU)$ the $j$ marginals satisfy:
\begin{equation}
M_{n}(\varphi_j)(v) = \int_\cV \varphi(u_j') M_n(v,dv') =  \int_\cU \varphi(u_j') M_{n,j}(u_j,du_j').\label{eq:marginal_m}
\end{equation}
We adopt the definition for $(v,\tilde{v})=((u_1,u_2),(\tilde{u}_1,\tilde{u}_2))$ of a sequence of Markov kernels 
$(\bar{M}_n)_{n\geq 1}$, $\bar{M}_n:\cV\times \cV\rightarrow\mathcal{P}(\cV)$
$$
\bar{M}_n((v,\tilde{v}),dv') := M_n((u_1,\tilde{u}_2),dv').
$$
In the main text $\cU=\bbR^d$, and in the references that follow $\cU$ should replace $\bbR^d$ in Assumptions \ref{asn:g} and \ref{asn:liperg}.

\subsection{Marginal Feynman-Kac Formula}

Given the above notations and defintions we define the $j-$marginal Feynman-Kac formulae:
$$
\gamma_{n,j}(du_n) = \int \prod_{p=0}^{n-1} G_p(u_p) \eta_{0,j}(du_0) \prod_{p=1}^n M_{p,j}(u_{p-1},du_p)
$$
with for $\varphi\in\mathcal{B}_b(\cU)$
$$
\eta_{n,j}(\varphi) = \frac{\gamma_{n,j}(\varphi)}{\gamma_{n,j}(1)}.
$$
One can also define the sequence of Bayes operators, for $\mu\in\mathcal{P}(\cU)$
$$
\Phi_{n,j}(\mu)(du) = \frac{\mu(G_{n-1}M_{n,j}(\cdot,du))}{\mu(G_{n-1})}\qquad n\geq 1.
$$
Recall that for $n\geq 1$, $\eta_{n,j} = \Phi_{n,j}(\eta_{n-1,j})$.

\subsection{Feynman-Kac Formulae for Multi-Level Particle Filters}

For $\mu\in\mathcal{P}(\cV)$ define for $u\in \cU$, $v\in \cV$:
\begin{eqnarray*}
G_{n,j,\mu}(u) & = & \frac{G_{n}(u)}{\mu_j(G_{n})} \\
\bar{G}_{n,\mu}(v) & = & G_{n,1,\mu}(u_1) \wedge G_{n,2,\mu}(u_2).
\end{eqnarray*}

Now for any sequence $(\mu_n)_{n\geq 0}$, $\mu_n\in\mathcal{P}(\cV)$, define the sequence of operators $(\bar{\Phi}_n(\mu_{n-1}))_{n\geq 1}$:
$$
\bar{\Phi}_n(\mu_{n-1})(dv_n) = 
$$
$$
\mu_{n-1}(\bar{G}_{n-1,\mu_{n-1}})\frac{\mu_{n-1}(\bar{G}_{n-1,\mu_{n-1}}M_n(\cdot,dv_n))}{\mu_{n-1}(\bar{G}_{n-1,\mu_{n-1}})} +(1-\mu_{n-1}(\bar{G}_{n-1,\mu_{n-1}}))\times
$$
$$
\mu_{n-1}\otimes
\mu_{n-1}
\Big(
\Big[
\frac{G_{n-1,1,\mu_{n-1}}-\bar{G}_{n-1,\mu_{n-1}}}{\mu_{n-1}(G_{n-1,1,\mu_{n-1}}-\bar{G}_{n-1,\mu_{n-1}})}\otimes
\frac{G_{n-1,2,\mu_{n-1}}-\bar{G}_{n-1,\mu_{n-1}}}{\mu_{n-1}(G_{n-1,2,\mu_{n-1}}-\bar{G}_{n-1,\mu_{n-1}})}
\Big]
\bar{M}_n(\cdot,dv_n)
\Big)
$$
Now define $\bar{\eta}_n := \bar{\Phi}_n(\bar{\eta}_{n-1})$ for $n\geq 1$, $\bar{\eta}_0=\eta_0$.

\begin{prop}\label{prop:marginal}
Let $(\mu_n)_{n\geq 0}$ be a sequence of probability measures on $\cV$ with $\mu_0=\eta_0$ and for each $j\in\{1,2\}$, $\varphi\in\mathcal{B}_b(\cU)$
$$
\mu_{n}(\varphi_j) = \eta_{n,j}(\varphi).
$$
Then:
$$
\eta_{n,j}(\varphi) = \bar{\Phi}_n(\mu_{n-1})(\varphi_j).
$$
In particular $\bar{\eta}_{n,j}=\eta_{n,j}$ for each $n\geq 0$.
\end{prop}

\begin{proof}
By assumption $M_n(\varphi_j)=M_{n,j}(\varphi)$, so
we have
\begin{eqnarray*}
\bar{\Phi}_n(\mu_{n-1})(\varphi_j) & = & \mu_{n-1}(\bar{G}_{n-1,\mu_{n-1}}M_{n,j}(\varphi)) + 
\mu_{n-1}
\Big(
\big[
G_{n-1,j,\mu_{n-1}}-\bar{G}_{n-1,\mu_{n-1}}
\big]
M_{n,j}(\varphi)
\Big)\\
& = & \mu_{n-1}(G_{n-1,j,\mu_{n-1}}M_{n,j}(\varphi))\\
& = & \eta_{n-1,j}(G_{n-1,j,\mu_{n-1}}M_{n,j}(\varphi))\\
& = & \Phi_{n,j}(\eta_{n-1,j})(\varphi) \\
& = & \eta_{n,j}(\varphi).
\end{eqnarray*}
\end{proof}

\begin{rem}
It is established that for any $\mu\in\mathcal{P}(\cV)$
\begin{equation}
\bar{\Phi}_n(\mu)(\varphi_j) = \Phi_{n,j}(\mu_{j})(\varphi).\label{eq:marginal_crazy_gives_standard}
\end{equation}
This property is very useful in subsequent calculations.
\end{rem}

The point of the proposition is that if one has a system that samples 
$\bar{\eta}_0$, $\bar{\Phi}_1(\bar{\eta}_0)$ and so on, that marginally, one has exactly the marginals
$\eta_{n,j}$ at each time point. In practice one cannot do this, but rather  
runs the following system:
$$
\Big(\prod_{i=1}^N \bar{\eta}_0(dv_0^i)\Big)\Big(\prod_{p=1}^n \prod_{i=1}^N \bar{\Phi}_{p}(\bar{\eta}^N_{p-1})(dv_p^i)\Big)
$$
which is exactly one pair of particle filters at a given level of the MLPF.

\section{Normalizing Constant}

First note that one can use the following 
$$
\prod_{p=0}^{n-1} \bar{\eta}_{p,j}^N(G_p)  
$$
to estimate $\gamma_{n,j}(1)$. It is now proven that this estimate is unbiased.

In particular, it will be shown that 
$$
(\prod_{p=0}^{n-1}\bar{\eta}_{p,j}^N(G_p))\bar{\eta}_{n,j}^N(\varphi)
$$
is an unbiased estimator of $\gamma_{n,j}(\varphi)$, and the above follows immediately. 
The proof is by induction and the result at step 
$0$ is clearly true. 
Now suppose it is true at step 
$n-1$ and consider the estimator above:
$$
\mathbb{E}\Big[\Big(\prod_{p=0}^{n-1}\bar{\eta}_{p,j}^N(G_p)\Big)\bar{\eta}_{n,j}^N(\varphi) \Big|\mathscr{F}_{n-1}^N\Big] = 
\Big(\prod_{p=0}^{n-1}\bar{\eta}_{p,j}^N(G_p)\Big) \mathbb{E}\Big[\bar{\eta}_{n,j}^N(\varphi)\Big|\mathscr{F}_{n-1}^N\Big]
$$
where $\mathscr{F}_{n-1}^N$ is the filtration generated by the particle system up-to time $n-1$.
Now, by the exchangeability of the particle system and \eqref{eq:marginal_crazy_gives_standard} :
$$
\mathbb{E}\Big[\bar{\eta}_{n,j}^N(\varphi)
\Big|\mathscr{F}_{n-1}^N\Big]
=
\bar{\Phi}_n(\bar{\eta}_{n-1}^N)(\varphi_j)
= \Phi_{n,j}(\bar{\eta}_{n-1,j}^N)(\varphi).
$$
So
$$
\mathbb{E}\Big[(\prod_{p=0}^{n-1}\bar{\eta}_{p,j}^N(G_p))\bar{\eta}_{n,j}^N(\varphi)\Big]
=
\mathbb{E}\Big[(\prod_{p=0}^{n-2}\bar{\eta}_{p,j}^N(G_p))\bar{\eta}_{n-1,j}^N(G_{n-1}M_{n,j}(\varphi))\Big].
$$
The induction hypothesis and standard results complete the proof.

\section{$\mathbb{L}_2-$Error}

The squared $\mathbb{L}_2-$Error (MSE)
is considered here.

\subsection{Results for the Filter}

Let
\begin{align}
\label{eq:beea}
B(n)  =& \Big(\sum_{p=0}^n \mathbb{E}[\{|u_{p,1}^1-u_{p,2}^1|\wedge 1 \}
^{2}]^{1/2}+\|\eta_{p,1}-\eta_{p,2}\|_{\textrm{tv}}+ \sum_{p=1}^n|||M_{p,1}-M_{p,2}|||\Big)^2.
\end{align}

\begin{theorem}\label{theo:filt_error}
Assume \ref{asn:g} and \ref{asn:liperg}. Then for any $n\geq 0$, $\varphi\in \mathcal{B}_b(\cU)\cap\textrm{\emph{Lip}}(\cU)$ there exist a $C(n,\varphi)<+\infty$
such that 
$$
\mathbb{E}\Bigg[\Bigg(
\frac{\bar{\eta}_n^N(G_{n,1}\varphi_1)}{\bar{\eta}_n^N(G_{n,1})} - 
\frac{\bar{\eta}_n^N(G_{n,2}\varphi_2)}{\bar{\eta}_n^N(G_{n,2})} -
\frac{\bar{\eta}_n(G_{n,1}\varphi_1)}{\bar{\eta}_n(G_{n,1})} +
\frac{\bar{\eta}_n(G_{n,2}\varphi_2)}{\bar{\eta}_n(G_{n,2})}
\Bigg)^2\Bigg] \leq 
\frac{C(n,\varphi)}{N}B(n).
$$

\end{theorem}

\begin{proof}
Follows directly from Lemma \ref{lem:kengo_lem} and similar calculations to the proof of Theorem \ref{theo:pred} for the term $\mathbb{E}\Big[\Big([\bar{\Phi}_n(\bar{\eta}_{n-1}^N)-\bar{\eta}_n](\varphi_1-\varphi_2)\Big)^2\Big]$.
\end{proof}

\begin{lem}\label{lem:only_one}
Assume \ref{asn:g} and \ref{asn:liperg}. 
Then for any $n\geq 1$, $\varphi\in \mathcal{B}_b(\cU)$ there exist a $C(n,\varphi)<+\infty$ such that
$$
\bigg|\mathbb{E}\bigg[\frac{\bar{\eta}_n^N(G_{n,1}\varphi_1)}{\bar{\eta}_n^N(G_{n,1})} 
-\frac{\bar{\eta}_n(G_{n,1}\varphi_1)}{\bar{\eta}_n(G_{n,1})}\bigg]\bigg|  +
\bigg|\mathbb{E}\bigg[\frac{\bar{\eta}_n^N(G_{n,2}\varphi_2)}{\bar{\eta}_n^N(G_{n,2})} - 
\frac{\bar{\eta}_n(G_{n,2}\varphi_2)}{\bar{\eta}_n(G_{n,2})}
\bigg]\bigg| \leq \frac{C(n,\varphi)}{N}.
$$
\end{lem}

\begin{proof}
The proof follows by using the bias result of Proposition 9.5.6 of \cite{delmoral1} (which holds in our context, see also Proposition \ref{prop:lp_conv}). 
\end{proof}

\begin{lem}\label{lem:inc_bias}
Assume \ref{asn:g} and \ref{asn:liperg}. Then for any $n\geq 1$, $\varphi\in \mathcal{B}_b(\cU)$ there exist a $C(n,\varphi)<+\infty$ such that
$$
\bigg|\mathbb{E}\bigg[\frac{\bar{\eta}_n^N(G_{n,1}\varphi_1)}{\bar{\eta}_n^N(G_{n,1})} 
- \frac{\bar{\eta}_n^N(G_{n,2}\varphi_2)}{\bar{\eta}_n^N(G_{n,2})}
-\frac{\bar{\eta}_n(G_{n,1}\varphi_1)}{\bar{\eta}_n(G_{n,1})}  +
\frac{\bar{\eta}_n(G_{n,2}\varphi_2)}{\bar{\eta}_n(G_{n,2})}
\bigg]\bigg| \leq C(n,\varphi)\frac{\sqrt{B(n)} 
}{N}.
$$
\end{lem}

\begin{proof}
For $p\le n$ and for $j=1,2$, let
\begin{align*}
Q_{p,n,j}(\varphi)(v_p)=\int G_{n,j}(u_{n,j})\varphi(u_{n,j})\prod_{p\le q< n}G_{q,j}(u_{q,j})M_{q,j}(u_{q,j},du_{q+1,j})\ (v_p=(u_{p,1},u_{p,2})). 
\end{align*}
Observe that
\begin{align}
	\overline{\eta}_p(Q_{p,n,1}(\varphi))-\overline{\eta}_p(Q_{p,n,2}(\varphi))&= \cO\left(\|\eta_{p,1}-\eta_{p,2}\|_{\textrm{tv}}+\sum_{p\le q<n}|||M_{p,1}-M_{p,2}|||\right)\nonumber\\
	&=\cO(\sqrt{B(n)}). 
	\label{eq:etaqb}
\end{align}
We prove the following by induction on $p\le n$:
\begin{align}\label{eq:etaq}
\bbE[(\bar{\eta}_p^N- \bar{\eta}_p)( Q_{p,n,1}(\varphi) - Q_{p,n,2}(\varphi))]  \le C(n,\varphi)\frac{\sqrt{B(n)}}{N}. 
\end{align}
The expectation is $0$ for $p=0$ by definition. Observe that
\[
\begin{split}
\bbE[(\bar{\eta}_{p+1}^N- \bar{\eta}_{p+1})(Q_{p+1,n,1}(\varphi) - Q_{p+1,n,2}(\varphi))] & =
\bbE[(\bar{\Phi}_{p+1}(\bar{\eta}_p^N)- \bar{\eta}_{p+1})( Q_{p+1,n,1}(\varphi) - Q_{p+1,n,2}(\varphi))]  \\
& = \bbE \bigg[ \frac{\bar{\eta}_p^N (Q_{p,n,1}(\varphi))}{\bar{\eta}_p^N (G_{p,1})}  
- \frac{\bar{\eta}_p^N(Q_{p,n,2}(\varphi))}{\bar{\eta}_p^N(G_{p,2})} \\
& -\frac{\bar{\eta}_p(Q_{p,n,1}(\varphi))}{\bar{\eta}_p(G_{p,1})}  +
\frac{\bar{\eta}_p(Q_{p,n,2}(\varphi))}{\bar{\eta}_p(G_{p,2})}
\bigg]. 
\end{split}
\]
Thus by taking $p=n$, the proof is complete if we can show (\ref{eq:etaq}). 
To prove (\ref{eq:etaq}), 
the departure point is Lemma \ref{lem:kengo_lem}, letting
$a=\bar{\eta}_p^N (Q_{p,n,1}(\varphi))$, $A=\bar{\eta}_p^N (G_{p,1})$,
$b=\bar{\eta}_p^N(Q_{p,n,2}(\varphi))$, $B=\bar{\eta}_p^N(G_{p,2})$,
  $c=\bar{\eta}_p(Q_{p,n,1}(\varphi))$, $C=\bar{\eta}_p(G_{p,1})$, 
  $d=\bar{\eta}_p(Q_{p,n,2}(\varphi))$, and $D=\bar{\eta}_p(G_{p,2})$.
Note the following estimates hold, by Thm. 3.1 of \cite{MR1843179}
\begin{equation}
\bbE[|a-c|^2]^{1/2}, ~ \bbE[|b-d|^2]^{1/2}, ~ \bbE[|A-C|^2]^{1/2}, ~ \bbE[|B-D|^2]^{1/2} = \cO(N^{-1/2}),
\label{eq:el2}
\end{equation}
as well as the following, by Lemma \ref{lem:only_one}
\begin{equation}
\bbE[a] -c , ~ \bbE[b]-d, ~ \bbE[A]-C, ~ \bbE[B] - D  = \cO(N^{-1}).
\label{eq:exp}
\end{equation}
Also, by (\ref{eq:etaqb}), 
\begin{equation}
c-d, ~ C-D = \cO\left(\sqrt{B(n)}\right).  
\label{eq:weak}
\end{equation}
Hence, by Equations \eqref{eq:weak} and \eqref{eq:exp} (noting that $c,C,d,D$ are not
random), the last 4 terms of Lemma \ref{lem:kengo_lem} are bounded by 
$\frac{C(n,\varphi)}{N}\sqrt{B(n)}$.

Observe that the first two terms of Lemma \ref{lem:kengo_lem} can be further decomposed into
\[
\begin{split}
\bbE \bigg[\frac{a-b-(c-d)}{A} - \frac{b [A-B-(C-D)]}{AB} \bigg] & = \frac{\bbE [a-b-(c-d)]}{C} - \frac{d\bbE [A-B-(C-D)]}{CD} 
\\ 
& - \bbE\bigg[ \frac{(A-C) [ a-b-(c-d)]}{AC} \bigg] \\
& - \bbE\bigg [ [A-B-(C-D)] \frac{(C-A)Db + (D-B)Ab + (b-d)AB}{ABCD} \bigg].
\end{split}
\]
The last two expectations above were $\cO(\sqrt{B(n)}/N)$ by applying Cauchy-Schwartz inequality  and 
using \eqref{eq:el2} and Theorem \ref{theo:filt_error}.  
Now, the first two terms above will be dealt with using the inductive hypothesis. Hence the proof is complete.

\end{proof}

\subsection{Results for the Predictor}

\begin{theorem}\label{theo:pred}
Assume \ref{asn:g} and \ref{asn:liperg}. 
Then for any $n\geq 0$, $\varphi\in \mathcal{B}_b(\cU)\cap\textrm{\emph{Lip}}(\cU)$ 
there exist a $C(n,\varphi)<+\infty$
such that 
$$
\mathbb{E}\Big[\Big([\bar{\eta}_n^N-\bar{\eta}_n](\varphi_1-\varphi_2)\Big)^2\Big] \leq \frac{C(n,\varphi)}{N}B(n). 
$$

\end{theorem}

\begin{proof}
The proof is by induction and clearly holds at 
step 0 by the Marcinkiewicz-Zygmund inequality (see e.g.~\cite{Cappe_2005}) so we proceed to the induction step.
Throughout $C$ is a constant whose value may change from line-to-line. Any important dependencies are given a function notation.

$$
\mathbb{E}\Big[\Big([\bar{\eta}_n^N-\bar{\eta}_n](\varphi_1-\varphi_2)\Big)^2\Big] \leq 
$$
\begin{equation}
2\mathbb{E}\Big[\Big([\bar{\eta}_n^N-\bar{\Phi}_n(\bar{\eta}_{n-1}^N)](\varphi_1-\varphi_2)\Big)^2\Big] +
2\mathbb{E}\Big[\Big([\bar{\Phi}_n(\bar{\eta}_{n-1}^N)-\bar{\eta}_n](\varphi_1-\varphi_2)\Big)^2\Big] \label{eq:standard_decomp}.
\end{equation}
Consider the two terms on the R.H.S.~of \eqref{eq:standard_decomp} separately.

\noindent\textbf{Term}: $\mathbb{E}\Big[\Big([\bar{\eta}_n^N-\bar{\Phi}_n(\bar{\eta}_{n-1}^N)](\varphi_1-\varphi_2)\Big)^2\Big]$.

Begin by conditioning on $\mathscr{F}_{n-1}^N$ and then apply the Marcinkiewicz-Zygmund inequality to yield that
$$
\mathbb{E}\Big[\Big([\bar{\eta}_n^N-\bar{\Phi}_n(\bar{\eta}_{n-1}^N)](\varphi_1-\varphi_2)\Big)^2\Big] \leq
$$
$$
\frac{C}{N}\Big(\mathbb{E}[|\varphi(u_{n,1}^1)-\varphi(u_{n,2}^1)|^2] + \mathbb{E}[|\bar{\Phi}_n(\bar{\eta}_{n-1}^N)(\varphi_1-\varphi_2)|^2]\Big) \leq
$$
\begin{equation}
\frac{C}{N}\Big(\mathbb{E}[\{|u_{n,1}^1-u_{n,2}^1|\wedge 1\}^2] + \mathbb{E}[|\bar{\Phi}_n(\bar{\eta}_{n-1}^N)(\varphi_1-\varphi_2)|^2]\Big) \label{eq:master_term1}
\end{equation}
where 
the final line follows since $\varphi\in \mathcal{B}_b(\cU)\cap\textrm{\emph{Lip}}(\cU)$.

Now by \eqref{eq:marginal_crazy_gives_standard}
\begin{eqnarray}
\bar{\Phi}_n(\bar{\eta}_{n-1}^N)(\varphi_1-\varphi_2) & = & \frac{\eta_{n-1,1}^N(G_{n-1}M_{n,1}(\varphi))-\eta_{n-1,2}^N(G_{n-1}M_{n,2}(\varphi))}{\eta_{n-1,1}^N(G_{n-1})} + \nonumber\\ & &
\frac{\eta_{n-1,2}^N(G_{n-1}M_{n,2}(\varphi))}{\eta_{n-1,1}^N(G_{n-1})\eta_{n-1,2}^N(G_{n-1})}[\eta_{n-1,2}^N(G_{n-1})-\eta_{n-1,1}^N(G_{n-1})]\label{eq:bayes_decomp}
\end{eqnarray}
Consider the first term on the R.H.S.~of \eqref{eq:bayes_decomp}. 
$$
\frac{\eta_{n-1,1}^N(G_{n-1}M_{n,1}(\varphi))-\eta_{n-1,2}^N(G_{n-1}M_{n,2}(\varphi))}{\eta_{n-1,1}^N(G_{n-1})} = \eta_{n-1,1}^N(G_{n-1})^{-1}[\eta_{n-1,1}^N(G_{n-1}M_{n,1}(\varphi))
$$
\begin{equation}
- \eta_{n-1,1}^N(G_{n-1}M_{n,2}(\varphi))
+ \eta_{n-1,1}^N(G_{n-1}M_{n,2}(\varphi)) -\eta_{n-1,2}^N(G_{n-1}M_{n,2}(\varphi))
]\label{eq:decomp1}
\end{equation}
Now we deal with $\eta_{n-1,1}^N(G_{n-1}M_{n,2}(\varphi)) -\eta_{n-1,2}^N(G_{n-1}M_{n,2}(\varphi))$ on the R.H.S.~of \eqref{eq:decomp1}.
$$
\eta_{n-1,1}^N(G_{n-1}M_{n,2}(\varphi)) -\eta_{n-1,2}^N(G_{n-1}M_{n,2}(\varphi)) = 
$$
$$
\frac{1}{N}\sum_{i=1}^N \Big\{[G_{n-1}(u_{n-1,1}^i)-G_{n-1}(u_{n-1,2}^i)]M_{n,2}(\varphi)(u_{n-1,1}^i)
+ 
$$
$$
G_{n-1}(u_{n-1,2}^i)[M_{n,2}(\varphi)(u_{n-1,1}^i)-M_{n,2}(\varphi)(u_{n-1,2}^i)]\Big\}.
$$
Then applying 
Assumptions \ref{asn:g} and \ref{asn:liperg} it follows that
\begin{equation}
|\eta_{n-1,1}^N(G_{n-1}M_{n,2}(\varphi)) -\eta_{n-1,2}^N(G_{n-1}M_{n,2}(\varphi))| \leq C(\varphi)\frac{1}{N}\sum_{i=1}^N\{|u_{n-1,1}^i-u_{n-1,2}^i|\wedge 1 \}
\label{eq:decomp2}
\end{equation}
Returning to \eqref{eq:decomp1} it follows that 
$$
|\eta_{n-1,1}^N(G_{n-1}M_{n,1}(\varphi)) - \eta_{n-1,1}^N(G_{n-1}M_{n,2}(\varphi))| \leq C(\varphi)|||M_{n,1}-M_{n,2}|||.
$$
Thus using 
Assumptions \ref{asn:g} and \ref{asn:liperg}
and noting \eqref{eq:decomp2} 
$$
\frac{\eta_{n-1,1}^N(G_{n-1}M_{n,1}(\varphi))-\eta_{n-1,2}^N(G_{n-1}M_{n,2}(\varphi))}{\eta_{n-1,1}^N(G_{n-1})} \leq
$$
\begin{equation}
C(\varphi)\Big(\frac{1}{N}\sum_{i=1}^N
\{|u_{n-1,1}^i-u_{n-1,2}^i|\wedge 1 \}
+|||M_{n,1}-M_{n,2}|||\Big)\label{eq:decomp3}.
\end{equation}
Returning to \eqref{eq:bayes_decomp} and the second term on the R.H.S.~it follows by the Lipschitz property of $G_{n-1}$ and the upper-bound on $\varphi$ and lower bound 
on $G_{n-1}$ that
$$
\frac{\eta_{n-1,2}^N(G_{n-1}M_{n,2}(\varphi))}{\eta_{n-1,1}^N(G_{n-1})\eta_{n-1,2}^N(G_{n-1})}[\eta_{n-1,2}^N(G_{n-1})-\eta_{n-1,1}^N(G_{n-1})] \leq
$$
\begin{equation}
C(\varphi)\frac{1}{N}\sum_{i=1}^N\{|u_{n-1,1}^i-u_{n-1,2}^i|\wedge 1 \}
\label{eq:decomp4}
\end{equation}
Recalling \eqref{eq:bayes_decomp} and noting \eqref{eq:decomp3}-\eqref{eq:decomp4} 
$$
\bar{\Phi}_n(\bar{\eta}_{n-1}^N)(\varphi_1-\varphi_2)  \leq 
C(\varphi)\Big(\frac{1}{N}\sum_{i=1}^N\{|u_{n-1,1}^i-u_{n-1,2}^i|\wedge 1 \}
+|||M_{n,1}-M_{n,2}|||\Big).
$$
Thus, on returning to \eqref{eq:master_term1} it follows that
$$
\mathbb{E}\Big[\Big([\bar{\eta}_n^N-\bar{\Phi}_n(\bar{\eta}_{n-1}^N)](\varphi_1-\varphi_2)\Big)^2\Big] \leq
$$
$$
\frac{C(\varphi)}{N}\Big(\mathbb{E}[\{|u_{n-1,1}^i-u_{n-1,2}^i|\wedge 1 \}
^2] +
\mathbb{E}\Big[\Big(\frac{1}{N}\sum_{i=1}^N\{|u_{n-1,1}^i-u_{n-1,2}^i|\wedge 1 \}
+|||M_{n,1}-M_{n,2}|||\Big)^2\Big]\Big) \leq 
$$
\begin{equation}
\frac{C(\varphi)}{N}\Big(\mathbb{E}[\{|u_{n-1,1}^i-u_{n-1,2}^i|\wedge 1 \}
^2] + \mathbb{E}[\{|u_{n-1,1}^i-u_{n-1,2}^i|\wedge 1 \}
^2] +|||M_{n,1}-M_{n,2}|||^2\Big)\label{eq:master1_bound}.
\end{equation}
The final equation follows from Jensen's inequality.

\noindent\textbf{Term}: $\mathbb{E}\Big[\Big([\bar{\Phi}_n(\bar{\eta}_{n-1}^N)-\bar{\eta}_n](\varphi_1-\varphi_2)\Big)^2\Big]$.

Application of Lemma \ref{lem:kengo_lem} to $[\bar{\Phi}_n(\bar{\eta}_{n-1}^N)-\bar{\eta}_n](\varphi_1-\varphi_2)$ allows one 
to treat the six terms independently, by the $C_2-$inequality. 
Denote the upper-bound in the induction hypothesis at time $n-1$
as $B_{n-1}(N)$ (omitting dependence on the function), to avoid complex notations.

\noindent\textbf{Term} 1: First 
$$
\mathbb{E}\Big[\Big(\frac{1}{\eta_{n-1,1}^N(G_{n-1})}\big(
\eta_{n-1,1}^N(G_{n-1}M_{n,1}(\varphi)) - \eta_{n-1,2}^N(G_{n-1}M_{n,2}(\varphi)) - 
$$
$$
\eta_{n-1,1}(G_{n-1}M_{n,1}(\varphi)) + \eta_{n-1,2}(G_{n-1}M_{n,2}(\varphi))
\big)\Big)^2\Big] \leq
$$
$$
C \mathbb{E}[(\eta_{n-1,1}^N(G_{n-1}M_{n,1}(\varphi)-G_{n-1}M_{n,2}(\varphi))-\eta_{n-1,1}(G_{n-1}M_{n,1}(\varphi)-G_{n-1}M_{n,2}(\varphi)))^2] + 
$$
$$
\mathbb{E}\Big[\Big([\bar{\eta}_{n-1}^N-\bar{\eta}_{n-1}]([G_{n-1}M_{n,2}(\varphi))]_1-
[G_{n-1}M_{n,2}(\varphi))]_2)
\Big)^2\Big].
$$
Application of Proposition \ref{prop:lp_conv} and the induction hypothesis yields the upper bound:
$$
\frac{C(n)|||M_{n,1}-M_{n,2}|||}{N} + B_{n-1}(N).
$$

\noindent\textbf{Term} 2: 
$$
\mathbb{E}\Big[\Big(\frac{\eta_{n-1,2}^N(G_{n-1}M_{n,2}(\varphi))}{\eta_{n-1,1}^N(G_{n-1})\eta_{n-1,2}^N(G_{n-1})}\big(
\eta_{n-1,1}^N(G_{n-1}) - \eta_{n-1,1}(G_{n-1}) - 
$$
$$
\eta_{n-1,2}^N(G_{n-1}) + \eta_{n-1,2}(G_{n-1})
\big)\Big) 
$$
$$
\leq C B_{n-1}(N).
$$

\noindent\textbf{Term} 3: 
By Proposition \ref{prop:lp_conv}
$$
\mathbb{E}\Big[\Big(\frac{1}{\eta_{n-1,1}^N(G_{n-1})\eta_{n-1,1}(G_{n-1})}\big(\eta_{n-1,1}-
\eta_{n-1,1}^N\big)(G_{n-1})
\big(\eta_{n-1,1}(G_{n-1}M_{n,1}(\varphi)) - 
$$
$$
\eta_{n-1,2}(G_{n-1}M_{n,2}(\varphi))\big)
\Big)^2\Big] \leq 
$$
$$
\frac{C(n)}{N}(|||M_{n,1}-M_{n,1}|||^2 + \|\eta_{n-1,1}-\eta_{n-1,2}\|_{\textrm{tv}}^2 + |||M_{n,1}-M_{n,1}|||
\|\eta_{n-1,1}-\eta_{n-1,2}\|_{\textrm{tv}}).
$$

\noindent\textbf{Term} 4: 
By Proposition \ref{prop:lp_conv}
$$
\mathbb{E}\Big[\Big(\frac{1}{\eta_{n-1,1}^N(G_{n-1})\eta_{n-1,2}^N(G_{n-1})}
\big(\eta_{n-1,2}^N(G_{n-1}M_{n,2}(\varphi))-\eta_{n-1,2}(G_{n-1}M_{n,2}(\varphi))
\big)
$$
$$
\big(\eta_{n-1,1}(G_{n-1})-\eta_{n-1,2}(G_{n-1})
\big)
\Big)^2\Big] \leq 
\frac{C(n)}{N}\|\eta_{n-1,1}-\eta_{n-1,2}\|_{\textrm{tv}}^2.
$$

\noindent\textbf{Term} 5: 
By Proposition \ref{prop:lp_conv}
$$
\mathbb{E}\Big[\Big(\frac{\eta_{n-1,2}(G_{n-1}M_{n,2}(\varphi))}{\eta_{n-1,1}(G_{n-1})\eta_{n-1,2}^N(G_{n-1})\eta_{n-1,2}(G_{n-1})}
\big(\eta_{n-1,2}^N(G_{n-1})-\eta_{n-1,2}(G_{n-1})
\big)
$$
$$
\big(\eta_{n-1,1}(G_{n-1})-\eta_{n-1,2}(G_{n-1})
\big)
\Big)^2\Big] \leq 
\frac{C(n)}{N}\|\eta_{n-1,1}-\eta_{n-1,2}\|_{\textrm{tv}}^2.
$$

\noindent\textbf{Term} 6: 
By Proposition \ref{prop:lp_conv}
$$
\mathbb{E}\Big[\Big(\frac{\eta_{n-1,2}(G_{n-1}M_{n,2}(\varphi))}{\eta_{n-1,1}^N(G_{n-1})\eta_{n-1,1}(G_{n-1})\eta_{n-1,2}^N(G_{n-1})}
\big(\eta_{n-1,1}^N(G_{n-1})-\eta_{n-1,1}(G_{n-1})
\big)
$$
$$
\big(\eta_{n-1,1}(G_{n-1})-\eta_{n-1,2}(G_{n-1})
\big)
\Big)^2\Big] \leq 
\frac{C(n)}{N}\|\eta_{n-1,1}-\eta_{n-1,2}\|_{\textrm{tv}}^2.
$$

Putting together the bounds on the terms 1-6 along with the bound on $\mathbb{E}\Big[\Big([\bar{\eta}_n^N-\bar{\Phi}_n(\bar{\eta}_{n-1}^N)](\varphi_1-\varphi_2)\Big)^2\Big]$ completes the proof.

\end{proof}

\begin{lem}\label{lem:kengo_lem}
Let $a,b,c,d,A,B,C,D\in\mathbb{R}$ with $A,B,C,D$ non-zero then:
$$
\frac{a}{A} - \frac{b}{B} - \left(\frac{c}{C} - \frac{d}{D} \right) = \frac{[a-b-(c-d)]}{A} -
\frac{b[A-B-(C-D)]}{AB} + \frac{1}{AC}[C-A][c-d] 
$$
$$
- \frac{1}{AB}(b-d)(C-D) +
\frac{d}{CBD}(B-D)(C-D) + \frac{d}{ACB}(A-C)(C-D).
$$
\end{lem}

\begin{prop}\label{prop:lp_conv}
Assume \ref{asn:g} and \ref{asn:liperg}.
Then for any $n\geq 0,p\geq 1$ there exists a $C(n,p)<+\infty$
such that for any $\varphi\in \mathcal{B}_b(\cU)$, $j\in\{1,2\}$, 
$$
\mathbb{E}[|[\eta_{n,j}^N-\eta_{n,j}](\varphi)|^p ]^{1/p} \leq \frac{C(n,p)\|\varphi\|}{\sqrt{N}}.
$$
\end{prop}

\begin{proof}
The proof is by induction and clearly holds at rank 0 by the Marcinkiewicz-Zygmund inequality so we proceed to the induction step.
Throughout $C$ is a constant whose value may change from line-to-line. Any important dependencies are given a function notation.

The triangle inequality provides
$$
\mathbb{E}[|[\eta_{n,j}^N-\eta_{n,j}](\varphi)|^p ]^{1/p} \leq
\mathbb{E}[|\eta_{n,j}^N(\varphi) -\bar{\Phi}_n(\bar{\eta}_{n-1}^N)(\varphi_j)|^p]^{1/p} +
\mathbb{E}[|\bar{\Phi}_n(\bar{\eta}_{n-1}^N)(\varphi_j) - \eta_{n,j}(\varphi)|^p]^{1/p}.
$$
For the first term on the R.H.S.~one can condition on $\mathscr{F}_{n-1}^N$ and then apply the Marcinkiewicz-Zygmund inequality to yield that
$$
\mathbb{E}[|\eta_{n,j}^N(\varphi) -\bar{\Phi}_n(\bar{\eta}_{n-1}^N)(\varphi_j)|^p]^{1/p} \leq \frac{C(n,p)\|\varphi\|}{\sqrt{N}}.
$$
For the second term on the R.H.S.~one has the decomposition (see \eqref{eq:marginal_crazy_gives_standard})
$$
\bar{\Phi}_n(\bar{\eta}_{n-1}^N)(\varphi_j) - \eta_{n,j}(\varphi) = 
$$
$$
\eta_{n-1,j}^N(G_{n-1})^{-1}[\eta_{n-1,j}^N(G_{n-1}M_{n,j}(\varphi))-\eta_{n-1,j}(G_{n-1}M_{n,j}(\varphi))] +
$$
$$
\frac{\eta_{n-1,j}(G_{n-1}M_{n,j}(\varphi))}{\eta_{n-1,j}^N(G_{n-1})\eta_{n-1,j}(G_{n-1})}[\eta_{n-1,j}(G_{n-1})-\eta_{n-1,j}^N(G_{n-1})].
$$
Then one can control $\mathbb{E}[|\bar{\Phi}_n(\bar{\eta}_{n-1}^N)(\varphi_j) - \eta_{n,j}(\varphi)|^p]^{1/p}$ via Minkowski, Assumptions \ref{asn:g} and \ref{asn:liperg}
and the induction hypothesis, to yield
$$
\mathbb{E}[|\bar{\Phi}_n(\bar{\eta}_{n-1}^N)(\varphi_j) - \eta_{n,j}(\varphi)|^p]^{1/p} \leq \frac{C(n,p)\|\varphi\|}{\sqrt{N}},
$$
and this allows one to conclude.
\end{proof}

\section{Estimates for Stochastic Diffusion Processes}

Consider the case of the diffusion example \eqref{eq:sde} of Section \ref{sec:setup}, 
with the multilevel kernel introduced in Subsection \ref{ssec:mlmc}.
Fix a level $l$, and for $x,y\in\mathbb{R}^d$, let $(X^x_1,X^y_2) \sim M((x,y),~\cdot~)$, 
i.e. $X^x_1$ is the solution at step $k_l$ of equation \eqref{eq:euler1step} with initial condition
$x$ and $X^y_2$ is the solution at step $k_{l-1}$ of equation \eqref{eq:euler1stepcoarse} 
with initial condition $y$.  
It is well-known that $\mathbb{E}[|X^x_1-X^x|^\kappa]^{1/\kappa}\le Ch_l^{1/2}$ for $\kappa>0$
(see for example \cite{MR3050787,KlPl92}),
where $X^x$ 
is also correlated to $X^x_1$, 
in the sense that the latter arises from a coarsening like \eqref{eq:euler1stepcoarse} except 
with an integration of the stochastic forcing $\xi(t)$ over the interval $h_l$.
Let us generalize this slightly and assume some method for which 
$\mathbb{E}[|X^x_1-X^x|^\kappa]^{1/\kappa}\le Ch_l^{\beta/2}$ for some $\beta>0$.

\begin{prop}\label{prop_EM}
Assume \ref{asn:diff}
and for any 
$x\in \mathbb{R}^d$ and $\kappa>0$,
that $\mathbb{E}[|X^x_1-X^x|^\kappa]^{1/\kappa}\le Ch_l^{\beta/2}$,
for some $\beta,C >0$. 
Now let $y\in\mathbb{R}^d$. Then there exists a $C' >0$ such that
\begin{align*}
\mathbb{E}\left[
|X_1^x - X_2^y|^\kappa\right]^{1/\kappa}\le C' (|x-y|+h_l^{\beta/2}).  
\end{align*}
\end{prop}

\begin{proof}
By the triangular inequality, it is sufficient to show
\begin{align*}
\mathbb{E}[|X^x_1-X^x|^\kappa]^{1/\kappa}&\le Ch_l^{\beta/2}\\
\mathbb{E}[|X^x-X^y|^\kappa]^{1/\kappa}&\le C'|x-y|, 
\end{align*}
The first inequality holds by assumption. 
Now note that Assumption \ref{asn:liperg} follows from Corollary V.11.7 of \cite{MR1780932} together with Gr\"onwall's inequality,
and the second estimate is immediate.  
\end{proof}

Note that this provides Assumption \ref{asn:mlrates}(ii).
For Euler 
the rate $\beta=1$ is well-known and may be found for example in \cite{MR3050787,KlPl92}.
Assume $M_{m,1}^l$ and $M_{m,2}^l$ are transition kernels corresponding to Euler-Maruyama scheme with grid sizes 
$h_l$ and $h_{l-1}$ respectively. 
Then, under the uniformly elliptic condition Assumption \ref{asn:diff}(i), by equation (2.4) of \cite{MR1843179}, 
\begin{align}
	|||M_{m,1}^l-M_{m,2}^l|||\le Ch_{l}. 
\label{eq:mbound}
\end{align}
This shows that the second term in Assumption \ref{asn:mlrates}(i) provides $\alpha=1$.
As for the first term of Assumption \ref{asn:mlrates}(i), preservation of the weak errror, 
the reader is referred to \cite{KlPl92, GrahamTalay}
where appropriate assumptions are detailed.
Now an inequality for predictors can be proven.

\begin{lem}\label{lem:DPJ}
Assume (\ref{asn:g}(i),\ref{asn:mlrates}(i)). 
For $l,m\in\mathbb{N}$, there exists $C>0$ such that
\begin{align*}
	\|\eta_{m,1}^l-\eta_{m,2}^l\|_{\textrm{\emph{tv}}}\le Ch^\alpha_{l}. 
\end{align*}
\end{lem}

\begin{proof}
	Let 
	\begin{align*}
	 	(H_{m,1}^l\varphi)(x)=\int M_{m,1}^l(x,dx^*)G_{m-1}(x)\varphi(x^*),\ 
		(H_{m,2}^l\varphi)(x)=\int M_{m,2}^l(x,dx^*)G_{m-1}(x)\varphi(x^*). 
	\end{align*}
	
	Then 
\begin{align*}
	\eta_{m,1}^l\varphi=\frac{\eta_{m-1,1}^lH_{m,1}^l\varphi}{\eta_{m-1,1}^lH_{m,1}^l 1},\ 
	\eta_{m,2}^l\varphi=\frac{\eta_{m-1,2}^lH_{m,2}^l\varphi}{\eta_{m-1,2}^lH_{m,2}^l 1}. 
\end{align*}
By definition,  $\eta_{0,1}^l=\eta_{0,2}^l$. Suppose that the claim holds for 
$0,1,\ldots, m-1$. Then 
\begin{align*}
	|\eta_{m,1}^l\varphi-\eta_{m,2}^l\varphi|=&
	\left|\frac{\eta_{m-1,1}^lH_{m,1}^l\varphi}{\eta_{m-1,1}^lH_{m,1}^l 1}-\frac{\eta_{m-1,2}^lH_{m,2}^l\varphi}{\eta_{m-1,2}^lH_{m,2}^l 1}\right|\\
	\le & 
	\frac{1}{\eta_{m-1,1}^lH_{m,1}^l 1}\left|\eta_{m-1,1}^lH_{m,1}^l\varphi-\eta_{m-1,2}^lH_{m,2}^l\varphi\right|\\
	&+
	\frac{\eta_{m-1,2}^lH_{m,2}^l\varphi}{\eta_{m-1,1}^lH_{m,1}^l 1\times \eta_{m-1,2}^lH_{m,2}^l 1}
	\left|\eta_{m-1,1}^lH_{m,1}^l1-\eta_{m-1,2}^lH_{m,2}^l1\right|. 
\end{align*}
By Assumption \ref{asn:g}(i), 
$c^{-1}\le \eta_{m-1,1}^lH_{m,1}^l 1, \eta_{m-1,2}^lH_{m,2}^l 1\le c$. 
Thus it is sufficient to show
\begin{align*}
\left|\eta_{m-1,1}^lH_{m,1}^l\varphi-\eta_{m-1,2}^lH_{m,2}^l\varphi\right|
\le C\|\varphi\| h_{l}^\alpha. 	
\end{align*}
However, the left-hand side of the above is dominated by 
\begin{align*}
&\left|\eta_{m-1,1}^lH_{m,1}^l\varphi-\eta_{m-1,2}^lH_{m,1}^l\varphi\right|
+
\left|\eta_{m-1,2}^lH_{m,1}^l\varphi-\eta_{m-1,2}^lH_{m,2}^l\varphi\right|\\
&\le 
\left(\|\eta_{m-1,1}^l-\eta_{m-1,2}^l\|_{tv}+|||M_{m,1}^l-M_{m,2}^l|||\right)\sup_{x,y}|G(y,x)|\|\varphi\|\le C\|\varphi\| h_{l}^\alpha. 
\end{align*}
where the second inequality follows from the induction assumption, and Assumptions \ref{asn:g}(i) and \ref{asn:mlrates}(i). 
Thus the claim follows by induction. 
\end{proof}

Let $I_{m,1}^l(k):=I_{m,1}^{l,k}$ and  $I_{m,2}^l(k):=I_{m,2}^{l,k}$. 
	For $m\ge 2$, let $S_m^l$ be the indices that choose the same ancestor in each resampling step, that is, 
	\begin{align*}
	S_m^l=\{k\in\{1,\ldots, N_l\}; &I_{m,1}^l(k)=I_{m,2}^l(k),I_{m-1,1}^l\circ I_{m,1}^l(k)=I_{m-1,2}^l\circ I_{m,2}^l(k),\cdots,\\
	&I_{1,1}^l\circ I_{2,1}^l\circ\cdots\circ I_{m,1}^l(k)=I_{1,2}^l\circ\cdots\circ I_{2,2}^l\circ I_{m,2}^l(k)\}. 
	\end{align*}
	For $m=1$, set $S_1^l=\{1,\ldots, N_l\}$. 
	Let 
\begin{align*}
\mathcal{F}_m^l=&\sigma\left(\left\{U_{p,1}^{l,k},U_{p,2}^{l,k}, \hat{U}_{p,1}^{l,k}, \hat{U}_{p,2}^{l,k}, I_{p,1}^l, I_{p,2}^l;p<m,k\le N_l\right\}\cup
\left\{U_{m,1}^{l,k},U_{m,2}^{l,k},k\le N_l\right\}\right),\\
\hat{\mathcal{F}}_m^l=&\sigma\left(\left\{U_{p,1}^{l,k},U_{p,2}^{l,k}, \hat{U}_{p,1}^{l,k}, \hat{U}_{p,2}^{l,k}, I_{p,1}^l, I_{p,2}^l;p<m,k\le N_l\right\}\cup
\left\{U_{m,1}^{l,k},U_{m,2}^{l,k},\hat{U}_{m,1}^{l,k},\hat{U}_{m,2}^{l,k},k\le N_l\right\}\right). 
\end{align*}

\begin{lem}\label{pairwise_difference}
For $\kappa>0$ and $m\in\mathbb{N}$,  there exists $C>0$ such that
\begin{align*} 
	\mathbb{E}\left[\frac{1}{N_l}\sum_{k\in S_{m-1}^l} |U_{m,1}^{l,k}-U_{m,2}^{l,k}|^\kappa\right]^{1/\kappa}\le Ch_{l}^{\beta/2}.
\end{align*}
\end{lem}

\begin{proof}
By Proposition \ref{prop_EM}, 
\begin{align*}
\mathbb{E}\left[\frac{1}{N_l}\sum_{k\in S_{m-1}^l} |U_{m,1}^{l,k}-U_{m,2}^{l,k}|^\kappa\right]^{1/\kappa}&=
\mathbb{E}\left[\frac{1}{N_l}\sum_{k\in S_{m-1}^l} \mathbb{E}\left[\left.|U_{m,1}^{l,k}-U_{m,2}^{l,k}|^\kappa\right|\hat{\mathcal{F}}_{m-1}^{l}\right]\right]^{1/\kappa}\\ 
&\le 
C\mathbb{E}\left[\frac{1}{N_l}\sum_{k\in S_{m-1}^l} \left\{|\hat{U}_{m-1,1}^{l,k}-\hat{U}_{m-1,2}^{l,k}|+h_{l}^{\beta/2}\right\}^\kappa\right]^{1/\kappa}. 
\end{align*}
Since $(a+b)^\kappa\le C(a^\kappa+b^\kappa)\ (a,b\ge 0)$, we have
\begin{align*}
\mathbb{E}\left[\frac{1}{N_l}\sum_{k\in S_{m-1}^l} |U_{m,1}^{l,k}-U_{m,2}^{l,k}|^\kappa\right]^{1/\kappa}&\le C\mathbb{E}\left[\frac{1}{N_l}\sum_{k\in S_{m-1}^l} |\hat{U}_{m-1,1}^{l,k}-\hat{U}_{m-1,2}^{l,k}|^\kappa\right]^{1/\kappa}+Ch_{l}^{\beta/2}\\
&=
C\mathbb{E}\left[\frac{1}{N_l}\sum_{k\in S_{m-1}^l} |U_{m-1,1}^{l,I_{m-1,1}^{l,k}}-U_{m-1,2}^{l,I_{m-1,2}^{l,k}}|^\kappa\right]^{1/\kappa}+Ch_{l}^{\beta/2}. 
\end{align*}
Note that $I_{m-1,1}^l=I_{m-1,2}^l$  for $k\in S_{m-1}^l$. 
The conditional distribution of $(U_{m-1,1}^{l,I_{m-1,1}^{l,k}},U_{m-1,2}^{l,I_{m-1,2}^{l,k}})\ (k\in S_{m-1}^l)$ given $S_{m-1}^l$  and $\mathcal{F}_{m-1}^l$ is
\begin{align*}
\frac{\sum_{k\in S_{m-2}^l} \frac{G_{m-1}(U_{m-1,1}^{l,k})}{\sum_{i=1}^{N_l}G_{m-1}(U_{m-1,1}^{l,i})}\wedge \frac{G_{m-1}(U_{m-1,2}^{l,k})}{\sum_{i=1}^{N_l}G_{m-1}(U_{m-1,2}^{l,i})}\delta_{(U_{m-1,1}^{l,k},U_{m-1,2}^{l,k})}
	}{\sum_{k\in S_{m-2}^l} \frac{G_{m-1}(U_{m-1,1}^{l,k})}{\sum_{i=1}^{N_l}G_{m-1}(U_{m-1,1}^{l,i})}\wedge \frac{G_{m-1}(U_{m-1,2}^{l,k})}{\sum_{i=1}^{N_l}G_{m-1}(U_{m-1,2}^{l,i})}}
	\le C\frac{1}{\sharp S_{m-2}^l}\sum_{k\in S_{m-2}^l}\delta_{(U_{m-1,1}^{l,k},U_{m-1,2}^{l,k})}
\end{align*}
The expected value of $\sharp S_{m-1}^l$ given $\mathcal{F}_{m-1}^l$ is
\begin{align*}
\mathbb{E}\left[\left.\frac{\sharp S_{m-1}^l}{N_l}\right|\mathcal{F}_{m-1}^l\right]=\sum_{k\in S_{m-2}^l} \frac{G_{m-1}(U_{m-1,1}^{l,k})}{\sum_{i=1}^{N_l}G_{m-1}(U_{m-1,1}^{l,i})}\wedge \frac{G_{m-1}(U_{m-1,2}^{l,k})}{\sum_{i=1}^{N_l}G_{m-1}(U_{m-1,2}^{l,i})}\le C \frac{\sharp S_{m-2}^l}{N_l}. 
\end{align*}
Therefore
\begin{align*}
&\mathbb{E}\left[\frac{1}{N_l}\sum_{k\in S_{m-1}^l} |U_{m-1,1}^{l,I_{m-1,1}^{l,k}}-U_{m-1,2}^{l,I_{m-1,2}^{l,k}}|^\kappa\right]\\
&=\mathbb{E}\left[\frac{1}{N_l}\sum_{k\in S_{m-1}^l} \mathbb{E}\left[\left.|U_{m-1,1}^{l,I_{m-1,1}^{l,k}}-U_{m-1,2}^{l,I_{m-1,2}^{l,k}}|^\kappa\right|S_{m-1}^l,\mathcal{F}_{m-1}^l\right]\right]\\
&=\mathbb{E}\left[\frac{\sharp S_{m-1}^l}{N_l}\left\{\frac{\sum_{k\in S_{m-2}^l} |U_{m-1,1}^{l,k}-U_{m-1,2}^{l,k}|^\kappa\frac{G_{m-1}(U_{m-1,1}^{l,k})}{\sum_{i=1}^{N_l}G_{m-1}(U_{m-1,1}^{l,i})}\wedge \frac{G_{m-1}(U_{m-1,2}^{l,k})}{\sum_{i=1}^{N_l}G_{m-1}(U_{m-1,2}^{l,i})}
	}{\sum_{k\in S_{m-2}^l} \frac{G_{m-1}(U_{m-1,1}^{l,k})}{\sum_{i=1}^{N_l}G_{m-1}(U_{m-1,1}^{l,i})}\wedge \frac{G_{m-1}(U_{m-1,2}^{l,k})}{\sum_{i=1}^{N_l}G_{m-1}(U_{m-1,2}^{l,i})}}\right\}\right]\\
&=\mathbb{E}\left[\mathbb{E}\left[\left.\frac{\sharp S_{m-1}^l}{N_l}\right|\mathcal{F}_{m-1}^l\right]\left\{\frac{\sum_{k\in S_{m-2}^l} |U_{m-1,1}^{l,k}-U_{m-1,2}^{l,k}|^\kappa\frac{G_{m-1}(U_{m-1,1}^{l,k})}{\sum_{i=1}^{N_l}G_{m-1}(U_{m-1,1}^{l,i})}\wedge \frac{G_{m-1}(U_{m-1,2}^{l,k})}{\sum_{i=1}^{N_l}G_{m-1}(U_{m-1,2}^{l,i})}
	}{\sum_{k\in S_{m-2}^l} \frac{G_{m-1}(U_{m-1,1}^{l,k})}{\sum_{i=1}^{N_l}G_{m-1}(U_{m-1,1}^{l,i})}\wedge \frac{G_{m-1}(U_{m-1,2}^{l,k})}{\sum_{i=1}^{N_l}G_{m-1}(U_{m-1,2}^{l,i})}}\right\}\right]\\
	&\le C\mathbb{E}\left[\frac{1}{N_l}\sum_{k\in S_{m-2}^l} |U_{m-1,1}^{l,k}-U_{m-1,2}^{l,k}|^\kappa\right].
\end{align*}
Thus the claim comes from induction. 
\end{proof}

\begin{lem}\label{small_probability}
There exists $C>0$ such that for $m\in\mathbb{N}$, 
\begin{align*}
1-\mathbb{E}\left[\frac{\sharp S_m^l}{N_l}\right]\le Ch_{l}^{\beta/2}. 
\end{align*}
\end{lem}

\begin{proof}
Note that 
\begin{align*}
1-\sum_{k=1}^{N_l}\frac{G_{m}(U_{m,1}^{l,k})}{\sum_{i=1}^{N_l} G_{m}(U_{m,1}^{l,i})}\wedge \frac{G_{m}(U_{m,2}^{l,k})}{\sum_{i=1}^{N_l} G_{m}(U_{m,2}^{l,i})}
=&\frac{1}{2}\sum_{k=1}^{N_l} \left|\frac{G_{m}(U_{m,1}^{l,k})}{\sum_{i=1}^{N_l} G_{m}(U_{m,1}^{l,i})}-\frac{G_{m}(U_{m,2}^{l,k})}{\sum_{i=1}^{N_l} G_{m}(U_{m,2}^{l,i})}\right|\\
\le &\frac{1}{2}\sum_{k\in S_{m-1}^l} \left|\frac{G_{m}(U_{m,1}^{l,k})}{\sum_{i=1}^{N_l} G_{m}(U_{m,1}^{l,i})}-\frac{G_{m}(U_{m,2}^{l,k})}{\sum_{i=1}^{N_l} G_{m}(U_{m,2}^{l,i})}\right|\\
&+\frac{1}{2}\sum_{k\notin S_{m-1}^l} \left|\frac{G_{m}(U_{m,1}^{l,k})}{\sum_{i=1}^{N_l} G_{m}(U_{m,1}^{l,i})}-\frac{G_{m}(U_{m,2}^{l,k})}{\sum_{i=1}^{N_l} G_{m}(U_{m,2}^{l,i})}\right|\\
	\le & C\frac{1}{N_l}\sum_{k\in S_{m-1}^l} |U_{m,1}^{l,k}-U_{m,2}^{l,k}|+C\left(1-\frac{\sharp S_{m-1}^l}{N_l}\right).
\end{align*}
Thus we have 
\begin{align*}
	\left(1-\mathbb{E}\left[\left.\frac{\sharp  S_m^l}{N_l}\right|\mathcal{F}_m^l\right]\right)
	=&\left\{1-\sum_{k=1}^{N_l}\frac{G_{m}(U_{m,1}^{l,k})}{\sum_{i=1}^{N_l} G_{m}(U_{m,1}^{l,i})}\wedge \frac{G_{m}(U_{m,2}^{l,k})}{\sum_{i=1}^{N_l} G_{m}(U_{m,2}^{l,i})}\right\}\\
	&+\sum_{k\notin S_{m-1}^l}\frac{G_{m}(U_{m,1}^{l,k})}{\sum_{i=1}^{N_l} G_{m}(U_{m,1}^{l,i})}\wedge \frac{G_{m}(U_{m,2}^{l,k})}{\sum_{i=1}^{N_l} G_{m}(U_{m,2}^{l,i})}\\
	\le & C\frac{1}{N_l}\sum_{k\in S_{m-1}^l} |U_{m,1}^{l,k}-U_{m,2}^{l,k}|+C\left(1-\frac{\sharp S_{m-1}^l}{N_l}\right). 
\end{align*}
The claim follows by induction. 
\end{proof}

\begin{theorem}
For $\kappa>1$ and $m\in\mathbb{N}$,  there exists $C>0$ such that
\begin{align*} 
	\mathbb{E}\left[\left(|U_{m,1}^{l,1}-U_{m,2}^{l,1}|\wedge 1\right)^\kappa\right]^{1/\kappa}\le Ch_{l}^{\beta/2\kappa}.
\end{align*}
\label{thm:main}
\end{theorem}

\begin{proof}
By Lemmas \ref{pairwise_difference} and \ref{small_probability}, 
\begin{align*} 
	\mathbb{E}\left[\left(|U_{m,1}^{l,1}-U_{m,2}^{l,1}|\wedge 1\right)^\kappa\right]
	&=
		\mathbb{E}\left[\frac{1}{N_l}\sum_{k=1}^{N_l}\left(|U_{m,1}^{l,k}-U_{m,2}^{l,k}|\wedge 1\right)^\kappa\right]\\
		&=
		\mathbb{E}\left[\frac{1}{N_l}\sum_{{k\in S_{m-1}^l}}\left(|U_{m,1}^{l,k}-U_{m,2}^{l,k}|\wedge 1\right)^\kappa\right]
		+	
		\mathbb{E}\left[\frac{1}{N_l}\sum_{{k\notin S_{m-1}^l}}\left(|U_{m,1}^{l,k}-U_{m,2}^{l,k}|\wedge 1\right)^\kappa\right]
		\\
		&\le Ch_{l}^{\kappa\beta/2}+Ch_{l}^{\beta/2}\le 2Ch_{l}^{\beta/2}. 
\end{align*}	
Thus the claim follows. 
\end{proof}

\begin{cor} \label{cor:finrate}
If $\gamma/\alpha<2$, then the bound of Theorem \ref{thm:mlpf1} is dominated by 
$$\sum_{l=0}^L \frac{C(m,\varphi)}{N_l} h_{l}^{\beta/2},$$
where $C(m,\varphi) = \max_{0\leq l \leq L} C_l(m,\varphi)$.
\end{cor}

\begin{proof}
First note that Theorem \ref{thm:main} provides a bound of $C_l(m,\varphi) h_l^{\beta/2}$
on the first term of $B_l(n)$ defined in 
\eqref{eq:beea}, 
and other
terms are bounded by $C_l(m,\varphi) h_l^{2\alpha}$.  
Recall that $2 \alpha \geq \beta$, as they are defined here.

Now, one must show that $ \sum_{l=0}^L\frac{\sqrt{B_l}}{N_l}\sum_{q=0\neq l}^L \frac{\sqrt{B_q}}{N_q}$
is higher order in comparison to $\sum_{l=0}^L \frac{B_l}{N_l} = \cO(\varepsilon^2)$.
Choosing $L$ and $K_L$ as described in Section \ref{ssec:mlmc} and the proof of Theorem \ref{thm:mlpf},
one has 
$$
 \sum_{l=0}^L \frac{\sqrt{B_l}}{N_l} \sum_{q=0\neq l}^L \frac{\sqrt{B_q}}{N_q} 
 \lesssim \varepsilon^4 K(\varepsilon)^{-2} \sum_{l=0}^L \sqrt{C_l} \sum_{q=0\neq l}^L \sqrt{C_q},
$$
where $C_l \propto h_l^{-\gamma}$ is the cost associated to the $l^{th}$ level.
Notice each of the two summations is $\cO(C_L) = \cO(\varepsilon^{-\gamma/2\alpha})$,
and $K(\varepsilon) = o(1)$.  Therefore, 
$$
 \sum_{l=0}^L \frac{\sqrt{B_l}}{N_l} \sum_{q=0\neq l}^L \frac{\sqrt{B_q}}{N_q} 
 \lesssim \varepsilon^{2} \varepsilon^{2-\gamma/\alpha},
 $$
and under the assumption that $\gamma/\alpha\leq2$ the proof is concluded.
\end{proof}

\end{document}